\documentclass[letterpaper,11pt]{article}
\usepackage[margin=1in]{geometry}  

\usepackage{amsmath}
\usepackage{amssymb}
\usepackage{amsthm}
\usepackage{array}
\usepackage{amssymb,graphicx}
\usepackage{mdwmath}
\usepackage{mdwtab}
\usepackage{eqparbox}
\usepackage{amsfonts}
\usepackage{tikz}
\usepackage{amsthm}
\usepackage{array}
\usepackage{mdwmath}
\usepackage{enumitem}
\usepackage{mdwtab}
\usepackage{eqparbox}
\usepackage{tikz}
\usepackage{tikz}
\usetikzlibrary{patterns}
\usepackage{latexsym}
\usepackage{cite}
\usepackage{amssymb}
\usepackage{bm}
\usepackage{bbm}
\usepackage{amssymb}
\usepackage{graphicx}
\usepackage{mathrsfs}
\usepackage{epsfig}
\usepackage{psfrag}
\usepackage{setspace}
\usepackage{hyperref}
\usepackage{algorithm}
\usepackage{algpseudocode}
\usepackage{stfloats}
\usepackage{listings}
\usepackage{verbatim}
\usepackage{pgfplots}
\usepackage{tcolorbox}
\usepackage{subfigure}

\newcommand{\Var}{\mathsf{Var}}
\newcommand{\1}{\mathbf{1}}
\allowdisplaybreaks

\usepackage{float}

\newtheorem{conjecture}{Conjecture} 
\newtheorem{example}{Example}

\newtheorem{theorem}{Theorem}

\newtheorem{lemma}{Lemma} 

\newtheorem{axiom*}{Axiom}

\newtheorem{observation}[theorem]{\textbf{Observation}}

\newtheorem{definition}{Definition}

\begin{document}
	
	\title{Concentration Inequalities for the Empirical Distribution of Discrete Distributions : Beyond the Method of Types}
	
\author{Jay Mardia, Jiantao Jiao, Ervin T\'anczos, Robert D. Nowak, Tsachy Weissman\thanks{Jay Mardia and Tsachy Weissman are with the Department of Electrical Engineering, Stanford University. Email: \{jmardia, tsachy\}@stanford.edu. Jiantao Jiao is with the Department of Electrical Engineering and Computer Sciences, University of California, Berkeley. Email: jiantao@berkeley.edu. Ervin T\'anczos and Robert D. Nowak are with the Department of Electrical and Computer Engineering, University of Wisconsin-Madison. Email: tanczos@wisc.edu, nowak@engr.wisc.edu.}}

	\maketitle
	
	\begin{abstract}
	We study concentration inequalities for the Kullback--Leibler (KL) divergence between the empirical distribution and the true distribution. Applying a recursion technique, we improve over the method of types bound uniformly in all regimes of sample size $n$ and alphabet size $k$, and the improvement becomes more significant when $k$ is large. We discuss the applications of our results in obtaining tighter concentration inequalities for $L_1$ deviations of the empirical distribution from the true distribution, and the difference between concentration around the expectation or zero. We also obtain asymptotically tight bounds on the variance of the KL divergence between the empirical and true distribution, and demonstrate their quantitatively different behaviors between small and large sample sizes compared to the alphabet size. 
	\end{abstract}
	
	\tableofcontents
	
	\begin{section}{Introduction and main results}\label{sec.Intro}
		
Concentration inequalities of empirical distributions play fundamental roles in probability theory, statistics, and machine learning. For example, the Kolmogorov--Smirnov goodness-of-fit test relies on the Dvoretzky--Kiefer--Wolfowitz--Massart inequality~\cite{massart1990tight} to control the significance level, the widely used Sanov's theorem~\cite[Theorem 11.4.1]{cover2012elements} is proved via the method of types concentration inequality, and the Vapnik--Chervonenkis~\cite{vapnik2015uniform} inequalites, chaining~\cite{talagrand2014upper} ideas, among others, provide foundational tools for statistical learning theory, and allow us to control the deviation of the empirical distribution from the true distribution under integral probability metrics~\cite{muller1997integral}. There have been works before that seek to improve well known concentration inequalities by incorporating additional distributional information. For example, \cite{berend2013concentration} and \cite[p 24]{raginsky2013concentration} both provide distribution dependent improvements to the Hoeffding bounds for the case of discrete distributions with support size $=2$. Our focus, however, will be on improving uniform concentration inequalities for discrete distributions with known support size that are distribution independent.

This paper focuses on obtaining concentration inequalities of the Kullback--Leibler (KL) divergence between the empirical distribution and the true distribution for discrete distributions. The KL divergence is not an integral probability metric, which makes it difficult to apply the VC inequality and chaining to obtain tight bounds. However, there is fundamental importance in understanding the behavior of the KL divergence. For example, what Sanov's theorem reveals is that as the sample size $n\to \infty$, under any distribution $P$, the probability of observing an empirical distribution $Q$ is characterized as $e^{-n D(Q\|P) + o(n)}$. Also, due to the Pinsker inequality~\cite[Chapter 11]{cover2012elements}, a concentration inequality for KL divergence also implies a concentration inequality for the total variation distance, which will also be used in this paper to provide an improved concentration inequality for the $L_1$ deviation. Concentration inequalities for the KL divergence between the estimated distribution and the real distribution in general exponential families were considered in~\cite{lai1988boundary,maillard2018boundary}. 

Concretely, we obtain $n$ i.i.d. samples $X^n = (X_1,X_2,\ldots,X_n)$ following distribution $P = (p_1,p_2,\ldots,p_k) \in \text{int}(\mathcal{M}_k)$\footnote{int stands for interior and this assumption simply means that all $p_i$ are positive. This assumption is without loss of generality because otherwise we could work with a smaller value of $k$.}, where $\mathcal{M}_k$ denotes the space of probability measures with alphabet size $k$. We are interested in concentration inequalities for the random variable 
		\begin{align*}
		D(\hat{P}_{n,k} \| P) \triangleq \sum_{i = 1}^k \hat{p}_i \log \frac{\hat{p}_i}{p_i},
		\end{align*}
		where $\hat{P}_{n,k} \triangleq (\hat{p}_1,\hat{p}_2,\ldots,\hat{p}_k)$ is the empirical distribution obtained from the $n$ samples. Throughout this paper, $\log$ is the natural logarithm.  
		
Probably the most well-known result on the concentration of $D(\hat{P}_{n,k} \| P)$ is due to the method of types~\cite[Lemma II.1]{csiszar1998method}, which is used in the proof of \cite[Sanov's Theorem 11.4.1]{cover2012elements}). It states that for any $\epsilon>0$, we have
		\begin{align} 
		\mathbb{P}\left( D(\hat{P}_{n,k} \| P) \geq \epsilon \right) & \label{eqn.typesgeneralbound} \leq {n + k-1 \choose k-1} e^{-n\epsilon} \nonumber \\
		\end{align}
		here ${x \choose y}$ is the binomial coefficient. 
		
		This bound is tight asymptotically when $k,\epsilon$ are fixed and $n\to \infty$ in the sense that
		\begin{align}
	 \lim_{n\to \infty} \frac{1}{n}	\log \left( \mathbb{P}\left( D(\hat{P}_{n,k} \| P) \geq \epsilon \right) \right)  = -\epsilon
		\end{align}
		as shown by~\cite[Theorem 11.4.1]{cover2012elements}, but it does not capture the correct dependence on $k$ in the non-asymptotic regime. In the modern era of big data, it is usually no longer a valid assumption that the sample size is significantly larger than the parameter dimension, and a clear understanding of the concentration inequalities in the non-asymptotic and large alphabet regime is becoming increasingly important.  
		
		\begin{example}\label{eg:TightFork=2}
			When $k = 2$, it is well known that the upper bound can be improved \cite[Remark (c) Theorem 2.2.3]{Dembo2010} to
			\begin{align}\label{eqn.binaryalphabet}
			\mathbb{P}\left( D(\hat{P}_{n,2} \| P) \geq \epsilon \right) & \leq 2 e^{-n\epsilon},
			\end{align} rather than the $(n+1)e^{-n\epsilon}$ bound from Equation \ref{eqn.typesgeneralbound}.
			This fact is a consequence of a union bound and the Sanov property of convex sets \cite[Theorem 1]{csiszar1984sanov} and is proved in Lemma \ref{lem.k=2Case} in Appendix \ref{sec.SumToIntegral} for the convenience of the reader. It is also clear that if we want a uniform bound that works for all $P$ and $\epsilon$, then this bound is tight. Consider $P = (\frac{1}{2},\frac{1}{2})$ and $\epsilon = \log 2$. Then $\mathbb{P}\left( D(\hat{P}_{n,2} \| P) \geq \epsilon \right) = \mathbb{P}\left(\hat{P}_{n,2} = (1,0) \cup \hat{P}_{n,2} = (0,1) \right) = 2 \cdot \frac{1}{2}^n = 2e^{-n\epsilon}$ where we have used the fact that $D\left(\hat{P}_{n,2}\|(\frac{1}{2},\frac{1}{2})\right)$ has maximum value $\log 2$ which is attained at the extremal points of $\mathcal{M}_2$.
		\end{example}
		
Following the asymptotic tightness of Sanov's theorem, we aim at bounds of the type
\begin{align}\label{eqn.boundweaimfor}
\mathbb{P}\left( D(\hat{P}_{n,k} \| P) \geq \epsilon \right) & \leq f(n,k) e^{-n\epsilon} = e^{n( \epsilon_{\text{thresh}} - \epsilon )}, 
\end{align}		
where $\epsilon_{\text{thresh}} = \frac{\log f(n,k)}{n}$. The threshold $\epsilon_{\text{thresh}}$ can be interpreted as the lower bound on $\epsilon$ such that (\ref{eqn.boundweaimfor}) becomes non-vacuous. 	Note that $\epsilon_{\text{thresh}}$ is defined with respect to a bound, but we will not introduce subscripts to denote this and let the corresponding upper bound be clear by context.
For the method of types bound, we have
\begin{align}\label{eqn.mtthreshold}
 \frac{k-1}{n}\log{\left(\frac{n+k-1}{k-1}\right)} \leq \epsilon_{\text{thresh}} = \frac{\log{ {n+k-1 \choose k-1}}}{n}.
\end{align}

Before we talk about improving $\epsilon_{\text{thresh}}$, however, we should see just how much we can hope to improve it. To this end, we observe that there also exists a lower bound for $\epsilon_{\text{thresh}}$ that follows from Equations~(\ref{eqn.UpperBoundpmom}) and ~(\ref{eqn.Meangeqk/n}) below. Indeed, since for non-negative random variable $X$ one has $\mathbb{E}[X^p] = \int_{0}^{\infty}pt^{p-1}\mathbb{P}(X\geq t)dt = \int_{0}^{\epsilon_{\text{thresh}}}pt^{p-1}\mathbb{P}(X\geq t)dt + \int_{\epsilon_{\text{thresh}}}^{\infty}pt^{p-1}\mathbb{P}(X\geq t)dt$, equation~(\ref{eqn.boundweaimfor}) implies that 
\begin{align}\label{eqn.UpperBoundpmom}
\mathbb{E}[D(\hat{P}_{n,k}\|P)^p] \leq \left(\epsilon_{\text{thresh}}\right)^p + \frac{p!}{n^p}. 
\end{align} 
It also follows from~\cite[Lemma 21]{jiao2017maximum} that when $n\geq 15k$ and $P$ is the uniform distribution, we have
\begin{align}
				\label{eqn.Meangeqk/n} \mathbb{E}[D(\hat{P}_{n,k}\|P)] \geq \frac{k-1}{2n}+\frac{k^2}{20n^2} - \frac{1}{12n^2},
\end{align}
and for any $n\geq 1$, \cite[Proposition 1]{paninski2003estimation} shows
				\begin{align}
				\label{eqn.Meanleqk/n} \mathbb{E}[D(\hat{P}_{n,k}\|P)] \leq \log \left(1+ \frac{k-1}{n}\right) \leq \frac{k-1}{n}, 
				\end{align}
which means that Equation ~(\ref{eqn.Meangeqk/n}) is tight up to constants. Together they imply that when $15k \leq n \leq C k$ and $P$ is uniform, $\epsilon_{\text{thresh}}$ is at least a constant because $\epsilon_{\text{thresh}} \geq \frac{k-3}{2n}$. So we can't, in general, hope to get an $\epsilon_{\text{thresh}}$ that is smaller than $\frac{k-3}{2n}$.

The contribution of this paper can be understood as obtaining uniformly smaller $\epsilon_{\text{thresh}}$ for all configurations of $n$ and $k$ compared with the method of types bound in~(\ref{eqn.mtthreshold}). These upper bounds on $\epsilon_{\text{thresh}}$ follow from the concentration inequality presented in Theorem~\ref{thm:main} stated in Section~\ref{sec.boundmain}. This inequality is complicated and we thus present slightly looser but easier to use bounds alongside, and spend Section ~\ref{sec.interpretmainthm} understanding how Theorem \ref{thm:main} compares to the method fo types bound.

Naturally, one may ask whether our Theorem~\ref{thm:main} provides the best non-asymptotic bound on the KL deviation. Although this question appears non-trivial, we hope to demonstrate through the following results on the variance of $D(\hat{P}_{n,k}\|P)$ that bounding $\mathbb{P}(D(\hat{P}_{n,k}\|P) \geq \epsilon)$ may not be the right question to ask in this context. 
			
\begin{theorem}\label{thm:VarIsk/nsq}
We have the following upper and lower bounds on the variance of $D(\hat{P}_{n,k}\|P)$. 
			\begin{itemize}
				\item There exists a universal constant $C$ such that for any $P \in \mathcal{M}_k$ and any $n$,
				\begin{align}
				\label{eqn.Varleqk/nsq} \Var\left[D(\hat{P}_{n,k}\|P)\right] \leq \min\left(\frac{6\cdot (3+\log k)^2}{n},C\frac{k}{n^2}\right).
				\end{align}
				\item For a fixed $k$, for any $ P \in \mathcal{M}_k$, asymptotically as $n$ goes to infinity,
				\begin{align}
				\label{eqn.Vargeqk/nsq} 2nD(\hat{P}_{n,k}\|P) \xrightarrow{\mathcal{D}} \chi^2_{k-1}.
				\\ 2(k-1) \leq \lim\limits_{n \rightarrow \infty} 4n^2 \Var\left[D(\hat{P}_{n,k}\|P)\right]. \nonumber
				\end{align}
				Here $\chi^2_{k-1}$ is the chi-square distribution with $k-1$ degrees of freedom and hence has variance $2(k-1)$. Notationally, $\xrightarrow{\mathcal{D}}$ means `converges in distribution'.
			\end{itemize}
		\end{theorem}
		Observe that the asymptotic lower bound of Theorem \ref{thm:VarIsk/nsq} implies that in the variance upper bound, the $C\frac{k}{n^2}$ terms is tight for large $n$. The variance bound $O\left( \frac{\log^2 k}{n} \right)$, is in fact tight when $n = 1$. We prove Theorem \ref{thm:VarIsk/nsq} in Section \ref{sec.VarThm}.

Combining Theorem~\ref{thm:VarIsk/nsq} with Equations~(\ref{eqn.Meangeqk/n}) and~(\ref{eqn.Meanleqk/n}), we observe that it requires at least $n\gg k$ to achieve vanishing expectation for $D(\hat{P}_{n,k}\|P)$ if $P$ is uniform in $\mathcal{M}_k$, but it only requires $n\gg (\log k)^2$ to achieve vanishing variance.

So, if $n\geq 15k$, and $P$ uniform, then we have
\begin{align*}
	\mathbb{E}[D(\hat{P}_{n,k}\|P)] = \Theta\left(\frac{k}{n}\right)
\end{align*}
and 
\begin{align*}
	\text{standard deviation of }D(\hat{P}_{n,k}\|P) = O\left(\frac{\sqrt{k}}{n}\right) \ll \Theta\left(\frac{k}{n}\right) = \mathbb{E}[D(\hat{P}_{n,k}\|P)].
\end{align*}
 In other words, in this regime the random variable $D(\hat{P}_{n,k}\|P)$ is concentrating very tightly around its expectation, and proving a concentration inequality of the type $\mathbb{P}(D(\hat{P}_{n,k}\|P) \geq \epsilon)$ fails to capture the different behavior of the expectation and the variance of $D(\hat{P}_{n,k}\|P)$. In this context, it may be more insightful to provide bounds for the centered concentration, i.e., 
\begin{align}
\mathbb{P}\left(|D(\hat{P}_{n,k}\|P) - \mathbb{E}[D(\hat{P}_{n,k}\|P)]|\geq \epsilon\right),
\end{align}
for which Theorem~\ref{thm:VarIsk/nsq} provides a bound via Chebyshev's inequality, but we suspect stronger (exponential) bounds are within reach. 

The main technique we use in the variance bound for Theorem~\ref{thm:VarIsk/nsq} is to break the quantity into smooth ($p_i > \frac{1}{n}$) and non-smooth ($p_i \leq \frac{1}{n}$) regions. We then apply polynomial approximation techniques inspired from \cite{paninski2003estimation} / \cite{braess2004bernstein} 
  when $p_i > \frac{1}{n}$ and then utilize the negative association properties of multinomial random variables~ \cite{joag1983negative}. Note that in \cite{paninski2003estimation} such a polynomial approximation technique provides a near-optimal scaling for $\mathbb{E}[D(\hat{P}_{n,k}\|P)]$ in the worst case, which suggests that it might be useful in analyzing the variance too. The reason we go through the effort of carefully handling dependent random variables in the proof of Theorem~\ref{thm:VarIsk/nsq} is that the Poissonized version of this problem gives an incorrect (and worse) scaling for the variance. We show the following lower bound on the variance in the Poissonized version of the problem.
		\begin{theorem}\label{thm.Poi}
		Let $\hat{P}_{n,k}^{\mathsf{Poi}} = \left(\hat{p}_1^{\mathsf{Poi}},\hat{p}_2^{\mathsf{Poi}},...,\hat{p}_k^{\mathsf{Poi}}\right)$ where each $\hat{p}_i^{\mathsf{Poi}}$ is independently distributed as $\frac{\mathsf{Poi}(np_i)}{n}$, where $\mathsf{Poi}(\lambda)$ is a Poisson random variable with parameter $\lambda$. Then for a fixed $k$, as $n \rightarrow \infty$
		\begin{align}\label{eqn.PoiVarLowBou}
			\sqrt{n}D(\hat{P}_{n,k}^{\mathsf{Poi}}\|P) \xrightarrow{\mathcal{D}} \mathcal{N}\left(0,1\right).
			\\ 1 \leq \lim\limits_{n \rightarrow \infty} n \Var\left[D(\hat{P}_{n,k}^{\mathsf{Poi}}\|P)\right]. \nonumber
		\end{align}
		\end{theorem}
		This shows that the Poissonized version of the problem cannot give us the right scaling for the variance ($\frac{k}{n^2}$ as from Theorem \ref{thm:VarIsk/nsq}) because of the asymptotic lower bound of $\frac{1}{n}$ in the Poissonized version.

		The concentration inequality presented in Theorem~\ref{thm:main} on KL divergence deviation can also be translated into a concentration inequality for the $L_1$ deviation via Pinsker's inequality. In fact, we use a strengthened version of Pinsker's inequality \cite[Theorem 2.2]{weissman2003inequalities} to obtain a bound that, to our knowledge, beats the best known concentration inequality for the $L_1$ distance between the empirical distribution and the true distribution for large $k$. We formally state these results in Section~\ref{sec.L1} in Lemma \ref{thm.L1Old} and Theorem \ref{thm.L1FromKl}. To compare known results with our results from Theorem \ref{thm.L1FromKl}, we plot Figures \ref{fig.L1Med_n} and \ref{fig.L1Small_n}.

The rest of the paper is organized as follows. We present the details of Theorem~\ref{thm:main} in Section~\ref{sec.boundmain}, compare the performance of this Theorem to the method of types in Section ~\ref{sec.interpretmainthm}, present the improved $L_1$ deviation inequality in Section~\ref{sec.L1}, and discuss future directions in Section~\ref{sec.FutDir}. The proofs of main results are collected in the Appendices~\ref{sec.VarThm} and~\ref{sec.maintheoremproof}. Appendix~\ref{sec.SumToIntegral} states some auxiliary lemmas and integrals that are used throughout this paper.

Notation:  We use the notation $a_\gamma \lesssim  b_\gamma$ to denote that there exists a universal constant $C$ such that $\sup_{\gamma } \frac{a_\gamma}{b_\gamma} \leq C$. Notation $a_\gamma \asymp b_\gamma$ is equivalent to $a_\gamma \lesssim  b_\gamma$ and $b_\gamma \lesssim  a_\gamma$. Notation $a_\gamma \gg b_\gamma$ means that $\liminf_\gamma \frac{a_\gamma}{b_\gamma} = \infty$, and $a_\gamma \ll b_\gamma$ is equivalent to $b_\gamma \gg a_\gamma$. The sequences $a_\gamma,b_\gamma$ are non-negative.

\subsection{The Kullback--Leibler concentration inequality} \label{sec.boundmain}

		Define 
		$c_0 = \pi, c_1 = 2, K_{-1}=1, d_0 = \max\{\pi,\frac{e^3}{2}\} = \frac{e^3}{2}$.
		
		\begin{align}\label{eqn.cmDefn}
		c_m  \triangleq \begin{cases}
		\frac{1 \times 3 \times 5 \ldots \times m-1}{2\times 4 \times 6 \times \ldots \times m} \cdot \pi & m\text{ is even}, m \geq 2 \\
		\frac{2\times 4 \times 6 \times \ldots \times m-1}{1 \times 3 \times 5 \ldots \times m} \cdot 2 & m\text{ is odd}, m \geq 3
		\end{cases}
		\end{align}
		
		\begin{align}\label{eqn.KmDefn}
		K_m \triangleq \prod_{j=0}^{m} c_j = \begin{cases}
		\frac{\pi(2\pi)^{\frac{m}{2}}}{2 \times 4 \times \ldots \times m} \leq \sqrt{\frac{\pi}{m}}(\sqrt{\frac{2\pi e}{m}})^m & m\text{ is even} 
		\\\frac{(2\pi)^{\frac{m+1}{2}}}{1 \times 3 \times \ldots \times m} \leq \sqrt{\frac{\frac{e^3}{2}}{m}}(\sqrt{\frac{2\pi e}{m}})^m & m\text{ is odd}
		\end{cases}
		\end{align}
		
		Observe that $c_m$ behaves as $\sqrt{\frac{2\pi}{m}}$ for large $m$ and that for all positive integers $m$, $K_m \leq \sqrt{\frac{d_0}{m}}(\sqrt{\frac{2\pi e}{m}})^m$.
		
		\begin{theorem}\label{thm:main}
			For all $n,k \geq 2$ and $P \in \mathcal{M}_k$, we have, for universal constants $C_0=(\frac{e^3}{2\pi}) \approx 3.1967$ and $C_1=\frac{3c_1}{c_2}\sqrt{\frac{d_0}{2\pi e}} \approx 2.9290$, the following, where $c_m$ and $K_m$ are defined as in~(\ref{eqn.cmDefn}) and (\ref{eqn.KmDefn}). 
			\begin{align}
			\label{eqn.ResultAllk} \mathbb{P}\left( D(\hat{P}_{n,k} \| P) \geq \epsilon \right) & \leq e^{-n\epsilon}\left[\frac{3c_1}{c_2}\sum_{i=0}^{k-2}K_{i-1} (\frac{e\sqrt{n}}{2\pi})^{i}\right] \leq e^{-n\epsilon}\left[\frac{3c_1}{c_2}\sqrt{\frac{d_0}{2\pi e}} \left(\sum_{i=1}^{k-2}\left(\sqrt{\frac{e^3n}{2\pi i}}\right)^{i}+1\right)\right]
			\end{align}
			The table that follows contains slightly looser but much more easily used and interpreted versions of the upper bound in (\ref{eqn.ResultAllk}).\\
			\begin{tabular}{ |p{7.5cm}|p{7.5cm}|  }
				\hline
				\multicolumn{2}{|c|}{More interpretable upper bounds for $\mathbb{P}\left( D(\hat{P}_{n,k} \| P) \geq \epsilon \right)$, $C_0=(\frac{e^3}{2\pi}) \approx 3.1967$} \\
				\hline
				\hline
				$$\text{Parameter Range}$$ & $$\text{Upper Bound}$$ \\
				\hline
				\hline
				$$3 \leq k \leq \sqrt{nC_0}+2$$ & $$C_1e\left(\sqrt{\frac{C_0n}{k}}\right)^ke^{-n\epsilon}$$ \\
				\hline
				$$3 \leq  k\leq \frac{nC_0}{e}+2$$ & $$C_1k\left(\sqrt{\frac{C_0n}{k}}\right)^ke^{-n\epsilon}$$ \\
				\hline
				$$\frac{nC_0}{e}+2 \leq k\leq nC_0+2$$ & $$C_1k e^{\frac{C_0n}{2e}}e^{-n\epsilon}$$ \\
				\hline
				$$k\geq nC_0+2$$ & $$C_1\left(nC_0 e^{\frac{C_0n}{2e}}+k\right)e^{-n\epsilon}$$ \\
				\hline
			\end{tabular}
		\end{theorem}
	The key technique we employ to prove Theorem~\ref{thm:main} is a recursive approach to reduce the problem with alphabet size $k$ to a problem with alphabet size $k-1$. Since we have good bounds for $k = 2$, we can induct from this case to obtain bounds for higher values of $k$.  We formalize this and present a proof in Appendix~\ref{sec.maintheoremproof}. 
	
	\begin{subsection}{Comparison with the method of types bound}\label{sec.interpretmainthm}
		In this section we compare the results of Theorem \ref{thm:main} to the Method of Types bound to better understand what improvements it yields.
		
		The contribution of this theorem can be understood as obtaining uniformly smaller $\epsilon_{\text{thresh}}$ for all configurations of $n$ and $k$ compared with the method of types bound in~(\ref{eqn.mtthreshold}). Our Theorem \ref{thm:main} implies the following $\epsilon_{\text{thresh}}$:
		
				\begin{tabular}{ |p{7.5cm}|p{7.5cm}|  }
			\hline
			\multicolumn{2}{|c|}{$\epsilon_{\text{thresh}}$ improvement using the upper bound from Theorem \ref{thm:main}, $C_0=(\frac{e^3}{2\pi}) \approx 3.1967$, $C_1  \approx 2.9290$} \\
			\hline
			\hline
			$$\text{Parameter Range}$$ & $$\epsilon_{\text{thresh}}(\text{Theorem \ref{thm:main}})$$ \\
			\hline
			\hline
			$$3 \leq k \leq \sqrt{nC_0}+2$$ & $$\frac{k\log \left(\sqrt{\frac{C_0n}{k}}\right)+\log (C_1e)}{n}  $$ \\
			\hline
			$$3 \leq  k\leq \frac{nC_0}{e}+2$$ & $$\frac{k\log \left(\sqrt{\frac{C_0n}{k}}\right)+\log (C_1k)}{n}$$ \\
			\hline
			$$\frac{nC_0}{e}+2 \leq k\leq nC_0+2$$ & $$\frac{\frac{nC_0}{2e}+\log (C_1k)}{n}$$ \\
			\hline
			$$k\geq nC_0+2$$ & $$\frac{\log \left(nC_0 e^{\frac{C_0n}{2e}}+k\right)+\log (C_1)}{n}$$ \\
			\hline
		\end{tabular}

		To better compare the improvement of $\epsilon_{\text{thresh}}$ in Theorem \ref{thm:main} with the method of types bound, we upper bound the ratio $\frac{\epsilon_{\text{thresh}}(\text{Theorem \ref{thm:main}})}{\epsilon_{\text{thresh}}(\text{method of types})}$ for several scalings of $n$, $k$ after letting $n \rightarrow \infty$.

		\begin{tabular}{ |p{5cm}|p{5cm}|p{5cm}|  }
			\hline
			\multicolumn{3}{|c|}{$\epsilon_{\text{thresh}}$ improvement using the upper bound from Theorem \ref{thm:main}, $C_0=(\frac{e^3}{2\pi}) \approx 3.1967$, $C_1  \approx 2.9290$} \\
			\hline
			\hline
			$$\text{Parameter Range}$$ & $$\frac{\epsilon_{\text{thresh}}(\text{Theorem \ref{thm:main}})}{\epsilon_{\text{thresh}}(\text{method of types})}$$ & $$\text{Upper bound on ratio}$$ \\
			\hline
			\hline
			$$k=o(n)$$ & $$\frac{k\log \left(\sqrt{\frac{C_0n}{k}}\right)+\log (C_1e)}{\log{ {n+k-1 \choose k-1}}}$$ & $$\frac{1}{2}$$\\
			\hline
			$$  k =  \frac{nC_0}{e}$$ & $$\frac{k\log \left(\sqrt{\frac{C_0n}{k}}\right)+\log (C_1k)}{\log{ {n+k-1 \choose k-1}}}$$ & $$\frac{1}{2}\cdot \frac{1}{\log \left(1+\frac{e}{C_0}\right)} \approx 0.8125$$\\
			\hline
			$$k = nC_0$$ & $$\frac{\frac{k}{2e}+\log (C_1k)}{\log{ {n+k-1 \choose k-1}}}$$ & $$\frac{1}{2e}\cdot \frac{1}{\log \left(1+\frac{1}{C_0}\right)} \approx 0.6758$$\\
			\hline
		\end{tabular}
			
		To compute the upper bound on the asymptotic ration in the table, we use the lower bound on $\epsilon_{\text{thresh}}(\text{method of types})$ from equation ~\ref{eqn.mtthreshold}.
		
		We notice that the constant improvement on the exponent $\epsilon_{\text{thresh}}$ provided by Theorem ~\ref{thm:main} can be viewed as a power function improvement on the tail probability. For $k = o(n)$, the asymptotic ratio of at least $\frac{1}{2}$ implies at least a square root improvement on the prefactor to the method of types bound.
		In fact, for regimes where $k \leq (\frac{nC_0}{4})^{\frac{1}{3}}$, we can obtain a better $\epsilon_{\text{thresh}}$ than the one obtained using the bound in Theorem \ref{thm:main} by letting $F = \{\hat{P}_{n,k} \in \mathcal{M}_k: D(\hat{P}_{n,k}\|P) \geq \epsilon)\}$ in Lemma \ref{lem:DiffSlopeBound} which gives a different bound on $\mathbb{P}\left(D(\hat{P}_{n,k}\|P) \geq \epsilon \right)$ and is proved in Appendix \ref{sec.SumToIntegral}. The $\epsilon_{\text{thresh}}$ one obtains from Lemma \ref{lem:DiffSlopeBound} is \[\frac{(k-1)\log{\left(2(k-1)\right)}}{n}\]. One can observe that this is smaller than the $\epsilon_{\text{thresh}}$ values in the table above when $2(k-1) \leq \sqrt{\frac{nC_0}{k}}$, which is satisfied when $k \leq (\frac{nC_0}{4})^{\frac{1}{3}}$. In particular, this shows that for any $k \leq (\frac{nC_0}{4})^{\frac{1}{3}}$ we get \[\frac{\epsilon_{\text{thresh}}(\text{method of types})}{\epsilon_{\text{thresh}}(\text{Lemma \ref{lem:DiffSlopeBound}})} = \Theta\left(\frac{\log\left(\frac{n+k}{k}\right)}{\log k}\right) .\] For any $k$ that is polylogarithmic in $n$, this is a superconstant improvement.

		To further illustrate the results, we plot the log of the upper bounds we obtain from Theorem \ref{thm:main} and from the method of types bound (along with the trivial upper bound one on probability) in Figures \ref{fig.nLarge}, \ref{fig.nMedium}, \ref{fig.nSmall}, and \ref{fig.nVSmall}.  In the numerical plots in these figures which accompany the cartoon plots, we have set $P$ to be the uniform distribution and used Monte Carlo simulations to calculate and plot the true probabilities.

		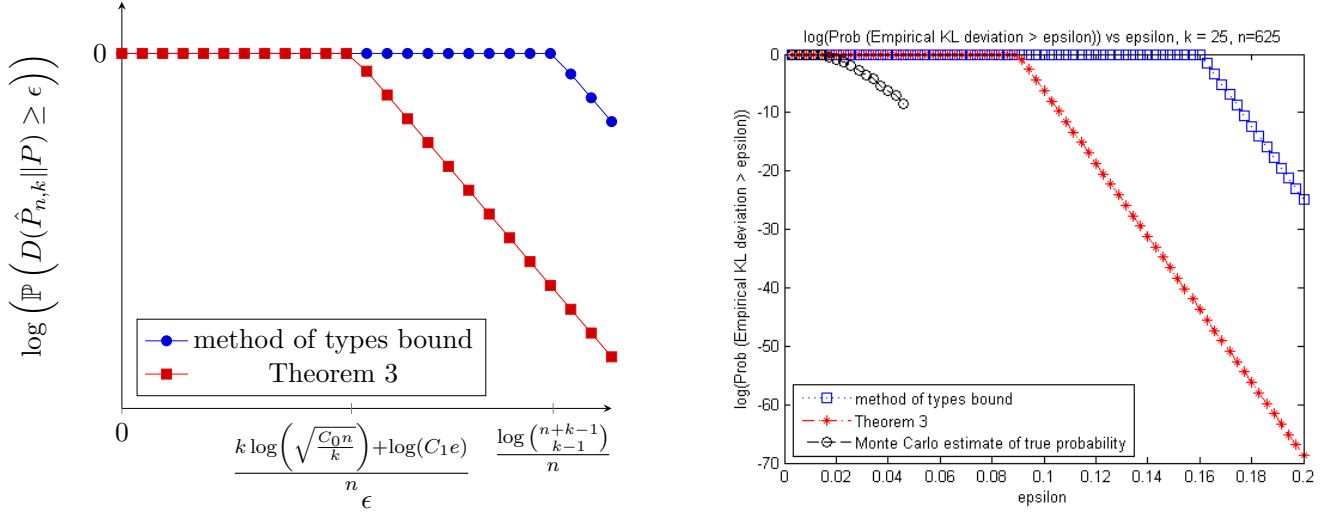
\begin{figure}[H]
			\begin{tikzpicture}[scale=.95]
			\pgfmathsetmacro{\d}{max((pi),((e)^3)/2)}
			\pgfmathsetmacro{\C}{(3*sqrt(\d / (2*pi*e)))}
			\pgfmathsetmacro{\D}{(((e)^3/(2*pi)))}
			\pgfmathsetmacro{\k}{31}
			\pgfmathsetmacro{\n}{1000.0}
			\pgfmathsetmacro{\epsold}{0.1365} 
			\pgfmathsetmacro{\epsnew}{(\k* ln(sqrt(\D* \n / \k))+ln(\C))/\n}

			\begin{axis}[
			legend pos=south west, axis y line = left, axis x line = bottom,
			title={$3\leq k \leq \sqrt{nC_0}+2$, $C_0=(\frac{e^3}{2\pi}) \approx 3.1967$, $C_1  \approx 2.9290$},
			x label style={at={(axis description cs:0.5,-0.1)},anchor=north},
			xlabel={$\epsilon$}, ylabel={$\log \left(\mathbb{P}\left( D(\hat{P}_{n,k} \| P) \geq \epsilon \right)\right)$}, domain = 0:(5*\k/\n), 
			enlarge y limits={rel=0.17}, xtick = \empty,extra x ticks = {0,
				\epsold ,\epsnew }, extra x tick labels = {$0$,$\frac{\log{ {n+k-1 \choose k-1}}}{n}$,$\frac{k\log \left(\sqrt{\frac{C_0n}{k}}\right)+\log (C_1e)}{n}$},ytick = {0},scaled y ticks = false]
			\addplot+{min(0,\n*\epsold-\n*x)};
			\addlegendentry{method of types bound}
			\addplot+{min(0,\n*\epsnew-\n*x)};
			\addlegendentry{Theorem \ref{thm:main}}

			\end{axis}
			\end{tikzpicture}
			\includegraphics[scale=0.60]{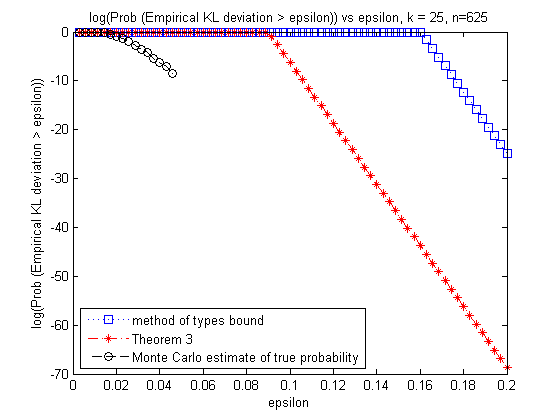}
			\caption{Cartoon plot and Numerical plot- $\log \left(\mathbb{P}\left( D(\hat{P}_{n,k} \| P) \geq \epsilon \right)\right)$ vs $\epsilon$ for large $n$\newline Observe that for large $n$, Theorem \ref{thm:main} is a significant improvement over the Method of types bound. Our cartoon plot on the left is corroborated with a numerical plot in which we fix $P$ to be uniform and also plot a Monte Carlo estimate of the true probabilities.}\label{fig.nLarge}
		\end{figure}

		\begin{figure}[H]
			\begin{tikzpicture}[scale=.95]
			\pgfmathsetmacro{\d}{max((pi),((e)^3)/2)}
			\pgfmathsetmacro{\C}{(3*sqrt(\d / (2*pi*e)))}
			\pgfmathsetmacro{\D}{(((e)^3/(2*pi)))}
			\pgfmathsetmacro{\k}{60}
			\pgfmathsetmacro{\n}{1000}
			\pgfmathsetmacro{\epsold}{.22763} 
			\pgfmathsetmacro{\epsnew}{((\k* ln(sqrt( \D*\n / \k))+ln(\C*\k))/\n}
			\begin{axis}[legend pos=south west, axis y line = left, axis x line = bottom,title={$\sqrt{nC_0}+2\leq  k \leq \frac{nC_0}{e}+2$, $C_0=(\frac{e^3}{2\pi}) \approx 3.1967$, $C_1  \approx 2.9290$},
			x label style={at={(axis description cs:0.5,-0.1)},anchor=north},
			xlabel={$\epsilon$}, ylabel={$\log \left(\mathbb{P}\left( D(\hat{P}_{n,k} \| P) \geq \epsilon \right)\right)$}, domain = 0:(4*\k/\n),
			enlarge y limits={rel=0.17}, xtick = \empty,extra x ticks = {0,\epsold ,\epsnew }, extra x tick labels = {$0$,$\frac{\log{ {n+k-1 \choose k-1}}}{n}$,\hspace{5pt}$\frac{k\log \left(\sqrt{\frac{C_0n}{k}}\right)+\log (C_1k)}{n}$},ytick = {0},scaled y ticks = false]
			
			\addplot+{min(0,\n*\epsold-\n*x)};
			\addlegendentry{method of types bound}
			\addplot+{min(0,\n*\epsnew-30-\n*x)};
			\addlegendentry{Theorem \ref{thm:main}}

			\end{axis}
			\end{tikzpicture}
			\includegraphics[scale=0.6]{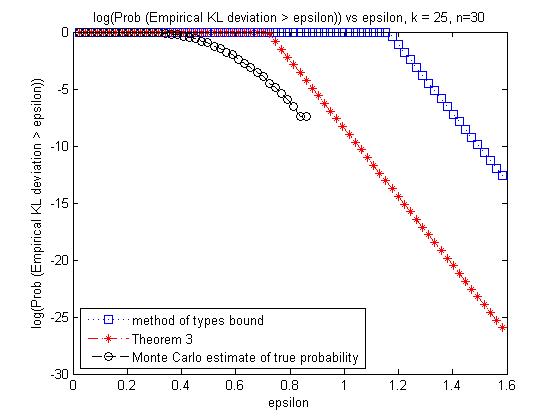}
			\caption{Cartoon plot and Numerical plot- $\log \left(\mathbb{P}\left( D(\hat{P}_{n,k} \| P) \geq \epsilon \right)\right)$ vs $\epsilon$ for medium $n$\newline When $n$ is medium sized, Theorem \ref{thm:main} is a significant improvement over the method of types bound and we can see it is much closer to the true probabilities than in the regime of Figure \ref{fig.nLarge}. Our cartoon plot on the left is corroborated with a numerical plot in which we fix $P$ to be uniform and also plot a Monte Carlo estimate of the true probabilities. }\label{fig.nMedium}
		\end{figure}
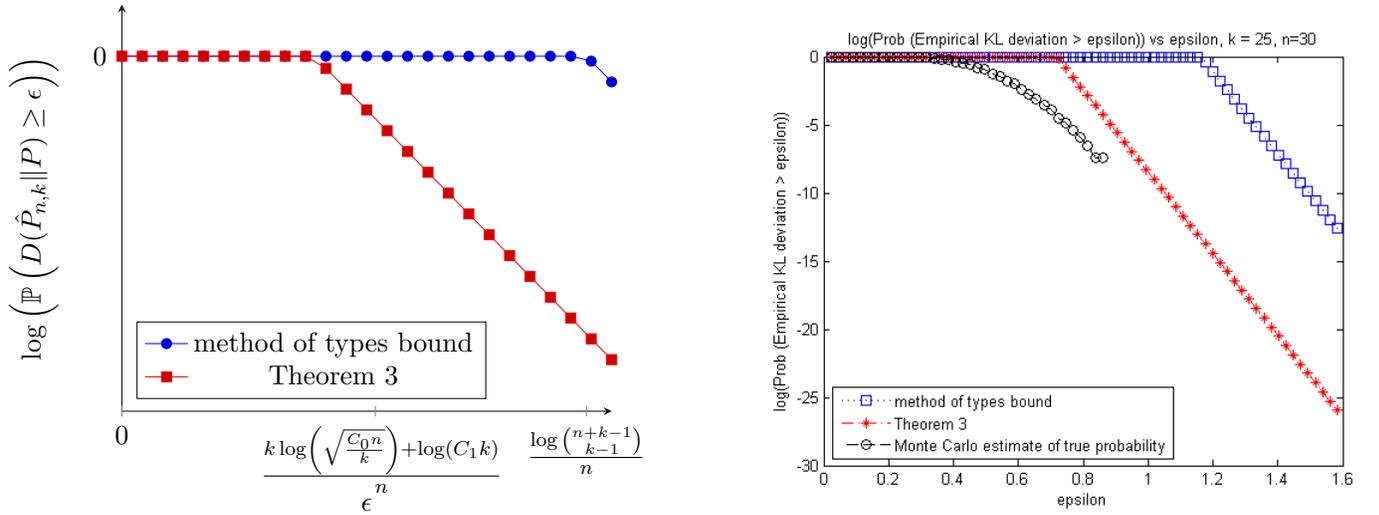

		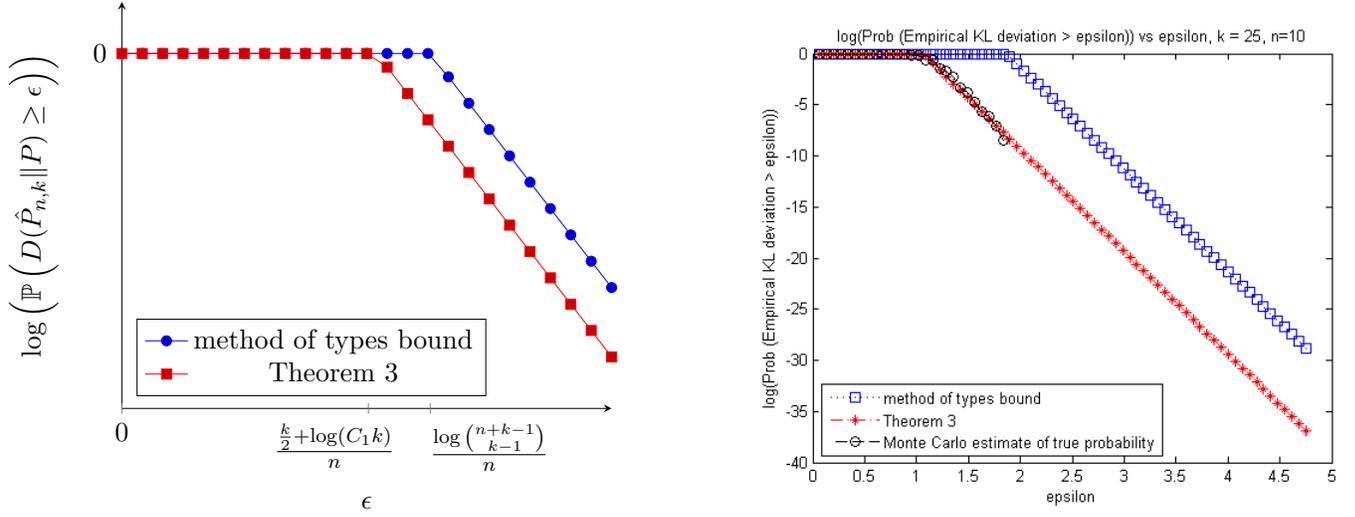
\begin{figure}[H]
			\begin{tikzpicture}[scale=0.95]
			\pgfmathsetmacro{\domainlen}{2}
			\pgfmathsetmacro{\d}{max((pi),((e)^3)/2)}
			\pgfmathsetmacro{\C}{(3*sqrt(\d / (2*pi*e)))}
			\pgfmathsetmacro{\D}{(((e)^3/(2*pi)))}
			\pgfmathsetmacro{\k}{1200}
			\pgfmathsetmacro{\n}{1000}
			\pgfmathsetmacro{\epsold}{1.511} 
			\pgfmathsetmacro{\epsnew}{((\k+ln(\C*\k) )/\n}
			\begin{axis}[legend pos=south west, axis y line = left, axis x line = bottom,title={$\frac{nC_0}{e}+2 \leq k\leq nC_0+2$, $C_0=(\frac{e^3}{2\pi}) \approx 3.1967$, $C_1  \approx 2.9290$},
			x label style={at={(axis description cs:0.5,-0.1)},anchor=north},
			xlabel={$\epsilon$}, ylabel={$\log \left(\mathbb{P}\left( D(\hat{P}_{n,k} \| P) \geq \epsilon \right)\right)$}, domain = 0:(\domainlen*\k/\n),
			enlarge y limits={rel=0.17}, xtick = \empty,extra x ticks = {0,\epsold ,\epsnew }, extra x tick labels = {$0$,\hspace{47pt}$\frac{\log{ {n+k-1 \choose k-1}}}{n}$,\hspace{-27pt}$\frac{\frac{k}{2}+\log (C_1k)}{n}$},ytick = {0},scaled y ticks = false]
			
			\addplot+{min(0,\n*\epsold-\n*x)};
			\addlegendentry{method of types bound}
			\addplot+{min(0,\n*\epsnew+40-\n*x)};
			\addlegendentry{Theorem \ref{thm:main}}

			\end{axis}
			\end{tikzpicture}	
			\includegraphics[scale=0.60]{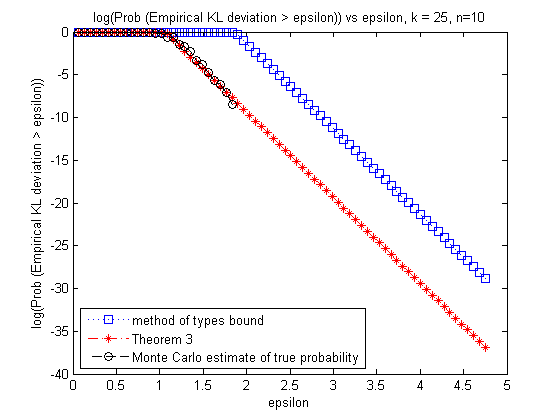}	
			\caption{Cartoon plot and Numerical plot- $\log \left(\mathbb{P}\left( D(\hat{P}_{n,k} \| P) \geq \epsilon \right)\right)$ vs $\epsilon$ for small $n$\newline For small $n$, Theorem \ref{thm:main} is much better than the Method of types bound. As demonstrated by our numerical plot, Theorem \ref{thm:main} and the true probabilities (for $P$ uniform and computed via a Monte Carlo simulation) are more or less the same. Hence in the regime of high dimensional distributions or very little data, Theorem \ref{thm:main} is essentially tight.}\label{fig.nSmall}
		\end{figure}
		
		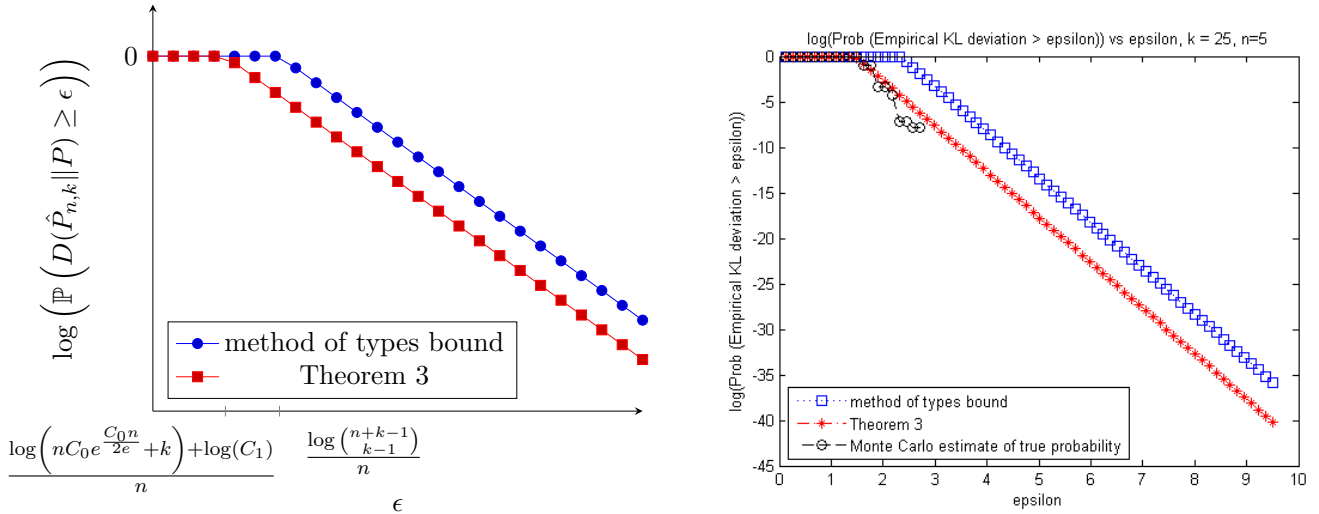
\begin{figure}[H]
			\begin{tikzpicture}[scale=0.95]
			\pgfmathsetmacro{\domainlen}{2}
			\pgfmathsetmacro{\d}{max((pi),((e)^3)/2)}
			\pgfmathsetmacro{\C}{(3*sqrt(\d / (2*pi*e)))}
			\pgfmathsetmacro{\D}{(((e)^3/(2*pi)))}
			\pgfmathsetmacro{\k}{100}
			\pgfmathsetmacro{\n}{20}
			\pgfmathsetmacro{\epsold}{2.586} 
			\pgfmathsetmacro{\epsnew}{((ln(\C*(\k + \n*\D*100))+20)/\n}
			\begin{axis}[legend pos=south west, axis y line = left, axis x line = bottom,title={$nC_0+2 \leq k$, $C_0=(\frac{e^3}{2\pi}) \approx 3.1967$, $C_1  \approx 2.9290$},
			x label style={at={(axis description cs:0.5,-0.1)},anchor=north},
			xlabel={$\epsilon$}, ylabel={$\log \left(\mathbb{P}\left( D(\hat{P}_{n,k} \| P) \geq \epsilon \right)\right)$}, domain = 0:(\domainlen*\k/\n),
			enlarge y limits={rel=0.17}, xtick = \empty,extra x ticks = {\epsold ,\epsnew }, extra x tick labels = {\hspace{67pt}$\frac{\log{ {n+k-1 \choose k-1}}}{n}$,\hspace{-67pt}$\frac{\log \left(nC_0 e^{\frac{C_0n}{2e}}+k\right)+\log (C_1)}{n}$},ytick = {0},scaled y ticks = false]
			
			\addplot+{min(0,\n*\epsold-\n*x)};
			\addlegendentry{method of types bound}
			\addplot+{min(0,\n*\epsnew-\n*x)};
			\addlegendentry{Theorem \ref{thm:main}}

			\end{axis}
			\end{tikzpicture}	
			\includegraphics[scale=0.60]{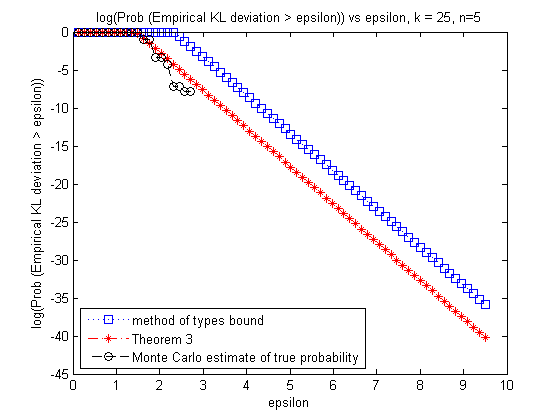}	
			\caption{Cartoon plot and Numerical plot- $\log \left(\mathbb{P}\left( D(\hat{P}_{n,k} \| P) \geq \epsilon \right)\right)$ vs $\epsilon$ for very small $n$\newline When $n$ is much smaller than $k$, qualitatively the picture is still the same as in Figure \ref{fig.nSmall}. Theorem \ref{thm:main}'s improvement over the method of types bound is enough that it numerically almost matches the Monte Carlo estimates of the true probabilities. (Plotted for $P$ uniform.)}\label{fig.nVSmall}
		\end{figure}

		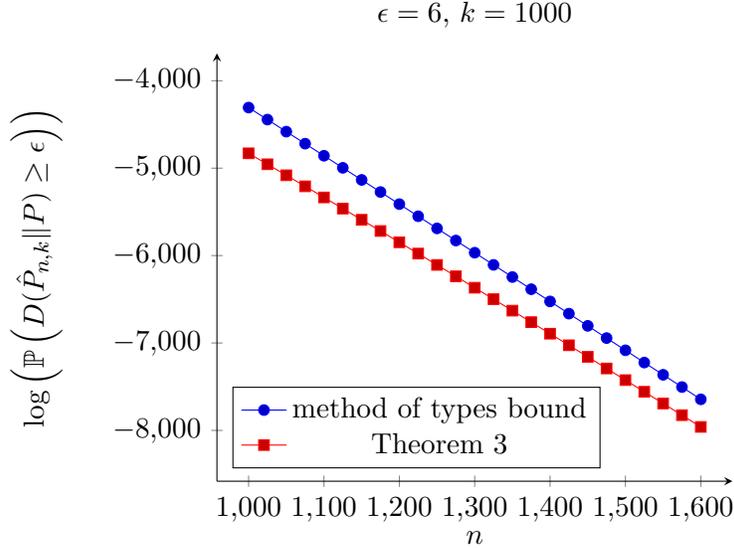
\begin{figure}[H]
			\centering
			\begin{tikzpicture}[scale=1]
			\pgfmathsetmacro{\d}{max((pi),((e)^3)/2)}
			\pgfmathsetmacro{\C}{(3*sqrt(\d / (2*pi*e)))}
			\pgfmathsetmacro{\D}{(((e)^3/(2*pi)))}
			\pgfmathsetmacro{\k}{1000}
			\pgfmathsetmacro{\epstest}{6}

			\begin{axis}[legend pos=south west, axis y line = left, axis x line = bottom,title={$\epsilon = \epstest$, $k = \k$},
			y label style={at={(axis description cs:-0.1,.5)},anchor=south},
			xlabel={$n$}, ylabel={$\log \left(\mathbb{P}\left( D(\hat{P}_{n,k} \| P) \geq \epsilon \right)\right)$}, domain = (1*\k):(1.6*\k),enlarge y limits={rel=0.17}, 
			enlarge x limits={rel=0.07}]

			\pgfmathdeclarefunction{func}{1}{%
				\pgfmathparse{%
					(#1<=(\k-2)/(\D))  * (ln((#1*\D)/(2*e))+((#1*\D)/(2*e))+ln(\k))  +%
					and(#1>(\k-2)/(\D)  ,#1<=e*(\k-2)/(\D)) * (ln(\k)+((#1*\D)/(2*e)))  +%
					and(#1>e*(\k-2)/(\D)  ,#1<=e*(\k-2)*(\k-2)/(\D)) * (ln(\k)+\k*ln((#1*\D)/(\k)))   +%
					(#1>e*(\k-2)*(\k-2)/(\D)) * (\k*ln((#1*\D)/(\k)))  %
				}%
			}

			\addplot+{min(0,-\epstest*x+(\k)*ln(e*(x+\k)/(\k)))};
			\addlegendentry{method of types bound};
			\addplot+{min(0,-\epstest*x+ln(\C)+func(x))};
			\addlegendentry{Theorem \ref{thm:main}};
			
			\end{axis}
			\end{tikzpicture}	
			\caption{Numerical plot- $\log \left(\mathbb{P}\left( D(\hat{P}_{n,k} \| P) \geq \epsilon \right)\right)$ vs $n$. This plot illustrates the fact that as $n$ decreases, Theorem \ref{thm:main} is a more and more significant improvement to the method of types bound. }\label{fig.ProbVsn}
		\end{figure}

	\end{subsection}

	\end{section}

	\begin{subsection}{Tightening $L_1$ deviation inequalities}\label{sec.L1}
			Having obtained results bounding the probability of a large KL deviation of the empirical distribution from the true distribution, we can now make use of Pinsker's inequality relating the $L_1$ distance between two distributions to the KL divergence between them to obtain bounds on the probability of $L_1$ deviation between the empirical and true distributions. We see that our bounds improve on the state-of-the-art bound from \cite[Theorem 2.1]{weissman2003inequalities} in the regime when $k \asymp n$. We first state some definitions and known results.
		\begin{align*}
		\|\hat{P}_{n,k}-P\|_1 \triangleq \sum \limits_{j=1}^{k} |\hat{p}_j-p_j|
		\end{align*}
		\begin{definition}\label{defn.PiP}
		Suppose $P$ is a discrete distribution with alphabet size $k$. Then, 
			\begin{align*}
			\pi_P \triangleq \max \limits_{A \subseteq [k]}\min \left(\mathbb{P}(A),1-\mathbb{P}(A)\right). 
			\end{align*}
		\end{definition}
		Note that $\pi_P\leq \frac{1}{2}$ for any $P$.
		\begin{definition}\label{defn.Phi}
			For $p \in [0,\frac{1}{2})$
			\begin{align*}
			\varphi(p) \triangleq \frac{1}{1-2p}\log\frac{1-p}{p},
			\end{align*}
			and by continuity set $\varphi(\frac{1}{2})=2$. 
		\end{definition}
		Observe that $\varphi(p) \geq 2$ for all $p \in [0,\frac{1}{2}]$. 
		\begin{lemma}{From \cite[Theorems 2.1 and 2.2]{weissman2003inequalities}, and with $\pi_P$, $\varphi(p)$ as defined in Definitions \ref{defn.PiP} and \ref{defn.Phi} }\label{thm.L1Old}
			\begin{itemize}
				\item Let $P$ be a probability distribution in $\mathcal{M}_k$. Then for all $n,k, \epsilon$-
				\begin{align*}
				\mathbb{P}\left(\|\hat{P}_{n,k}-P\|_1 \geq \epsilon \right) \leq \left(2^k-2\right)e^{-\frac{n\varphi(\pi_P)\epsilon^2}{4}}.
				\end{align*}
				\item Let $P,Q$ be a probability distributions in $\mathcal{M}_k$. Then we have the following strengthened Pinsker inequality
				\begin{align*}
				\|Q-P\|_1 \leq 2\sqrt{\frac{D(Q\|P)}{\varphi(\pi_P)}}.
				\end{align*}
			\end{itemize}
			
		\end{lemma}
		Now, with the slightly strengthened version of the Pinkser inequality from Lemma \ref{thm.L1Old} we can get the inequality \[\mathbb{P}\left(\|\hat{P}_{n,k}-P\|_1 \geq \epsilon \right) \leq \mathbb{P}\left(2\sqrt{\frac{D(\hat{P}_{n,k}\|P)}{\varphi(\pi_P)}} \geq \epsilon \right),\] which combined with Theorem \ref{thm:main} allows us to obtain the following bounds on the $L_1$ deviation probability.
		\begin{theorem}\label{thm.L1FromKl}
			For all $n,k \geq 2$ and $P \in \mathcal{M}_k$, we have, for universal constants $C_0=(\frac{e^3}{2\pi}) \approx 3.1967$ and $C_1=\frac{3c_1}{c_2}\sqrt{\frac{d_0}{2\pi e}} \approx 2.9290$, the following, where $c_m$ and $K_m$ are defined as in Equations \ref{eqn.cmDefn} and Equation \ref{eqn.KmDefn} and with $\pi_P$, $\varphi(p)$ as defined in Definitions \ref{defn.PiP} and \ref{defn.Phi}
			\begin{align}
			\mathbb{P}\left(\|\hat{P}_{n,k}-P\|_1 \geq \epsilon \right) & \leq e^{-\frac{n\varphi(\pi_P)\epsilon^2}{4}}\left[\frac{3c_1}{c_2}\sum_{i=0}^{k-2}K_{i-1} (\frac{e\sqrt{n}}{2\pi})^{i}\right] \\
			& \leq e^{-\frac{n\varphi(\pi_P)\epsilon^2}{4}}\left[\frac{3c_1}{c_2}\sqrt{\frac{d_0}{2\pi e}} \left(\sum_{i=1}^{k-2}\left(\sqrt{\frac{e^3n}{2\pi i}}\right)^{i}+1\right)\right]. 
			\end{align}
		\end{theorem}
		We can make these bounds more interpretable in the same way as in Theorem \ref{thm:main} in the table following it. Here we note the regimes of $k$ and $n$ for which Theorem \ref{thm.L1FromKl} actually provides better bounds than Lemma \ref{thm.L1Old} and compare the two in the following table.\\ 
		\begin{tabular}{ |p{5cm}|p{5cm}|p{5cm}|  }
			\hline
			\multicolumn{3}{|c|}{Comparing upper bounds for $\mathbb{P}\left(\|\hat{P}_{n,k}-P\|_1 \geq \epsilon \right)$ from Lemma \ref{thm.L1Old} and Theorem \ref{thm.L1FromKl}, $C_0=(\frac{e^3}{2\pi}) \approx 3.1967$} \\
			\hline
			\hline
			$$\text{Parameter Range}$$ & $$\text{Theorem \ref{thm.L1FromKl}}$$ & $$\text{Lemma \ref{thm.L1Old}\cite[Theorem 2.1]{weissman2003inequalities}}$$ \\
			\hline
			\hline
			$$\frac{nC_0}{4}+2 \leq  k\leq \frac{nC_0}{e}+2$$ & $$C_1k\left(\sqrt{\frac{C_0n}{k}}\right)^ke^{-\frac{n\varphi(\pi_P)\epsilon^2}{4}}$$ & $$\left(2^k-2\right)e^{-\frac{n\varphi(\pi_P)\epsilon^2}{4}}$$ \\
			\hline
			$$\frac{nC_0}{e}+2 \leq k$$ & $$C_1k e^{\frac{C_0n}{2e}}e^{-\frac{n\varphi(\pi_P)\epsilon^2}{4}}$$ & $$\left(2^k-2\right)e^{-\frac{n\varphi(\pi_P)\epsilon^2}{4}}$$ \\
			\hline
		\end{tabular}
		
\vspace{10pt}		
		 Hence for large enough $k$, Theorem \ref{thm.L1FromKl} outperforms Lemma \ref{thm.L1Old} in the regimes listed above. Of course, keeping in mind that $L_1$ distance between probability distributions takes a maximum value of $2$, this improvement is only meaningful if either of these upper bounds does better than the trivial upper bound of $1$ on any probability. However, since $\varphi(\pi_P) \geq 2$ for all $P \in \mathcal{M}_k$, this happens even for $n \asymp k$. We illustrate this and the difference between Lemma \ref{thm.L1Old} and Theorem \ref{thm.L1FromKl} in Figures \ref{fig.L1Med_n} and \ref{fig.L1Small_n}. Observe that the smaller $n$ becomes, the more significantly improved our results become compared to known bounds and approach the true probabilities (as computed using Monte Carlo simulations). For the Monte Carlo simulation of true probabilities, we have set $P$ to be the uniform distribution.

		\begin{figure}[H]
			\begin{tikzpicture}[scale=0.95]
			\pgfmathsetmacro{\d}{max((pi),((e)^3)/2)}
			\pgfmathsetmacro{\C}{(3*sqrt(\d / (2*pi*e)))}
			\pgfmathsetmacro{\D}{(((e)^3/(2*pi)))}
			\pgfmathsetmacro{\k}{1250}
			\pgfmathsetmacro{\n}{1000}
			\pgfmathsetmacro{\epsold}{((\k-1)*ln(\n+1))/\n}
			\pgfmathsetmacro{\epsnew}{((\k* ln(sqrt( \D*\n / \k))+ln(\C*\k))/\n}
			
			\pgfmathsetmacro{\plotmarkernew}{sqrt(2*\epsnew)}
			\pgfmathsetmacro{\plotmarkerold}{sqrt(\k * 2 * ln(2) / \n)}

			\begin{axis}[legend pos=south west, axis y line = left, axis x line = bottom,title={$\frac{nC_0}{4}+2\leq  k \leq \frac{nC_0}{e}+2$, $C_0=(\frac{e^3}{2\pi}) \approx 3.1967$, $C_1  \approx 2.9290$}, ylabel={$\log \left(\mathbb{P}\left(\|\hat{P}_{n,k}-P\|_1 \geq \epsilon \right))\right)$},x label style={at={(axis description cs:0.5,-0.15)},anchor=north}, xlabel={$\epsilon$} , domain = 0:2,
			enlarge y limits={rel=0.17}, xtick = \empty,extra x ticks = {0,\plotmarkernew,\plotmarkerold,2}, extra x tick labels = {$0$,\hspace{-3cm}$\sqrt{\frac{4\left(k\log \left(\sqrt{\frac{nC_0}{k}}\right)+\log(C_1k)\right)}{n\varphi(\pi_P)}}$,\hspace{45pt}$\sqrt{\frac{4k\log (2)}{n\varphi(\pi_P)}}$,$2$},ytick = {0},scaled y ticks = false]

			\addplot+{min(0,\k*ln(2)-\n/2*(x*x)};
			\addlegendentry{Lemma \ref{thm.L1Old} \cite[Theorem 2.1]{weissman2003inequalities}}
			\addplot+{min(0,\n*\epsnew-\n/2*(x*x))};
			\addlegendentry{Theorem \ref{thm.L1FromKl}}
			
			\end{axis}
			\end{tikzpicture}
			\includegraphics[scale=0.60]{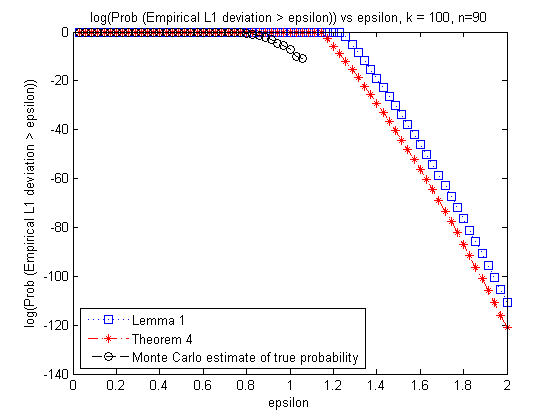}
			\caption{Cartoon and numerical plots- $\log \left(\mathbb{P}\left(\|\hat{P}_{n,k}-P\|_1 \geq \epsilon \right)\right)$ vs $\epsilon$ for medium $n$. \newline When $n$ is comparable to $k$, we plot cartoon versions of the previously known bounds from Lemma \ref{thm.L1Old} and our bounds from Theorem \ref{thm.L1FromKl}. The latter shows a non-trivial but slight improvement, and we corroborate this plot with a numerical plot for $k=100$, $n=90$ with $P$ uniform (and hence $\varphi(\pi_P) = 2$). In the numerical figure we also plot a Monte Carlo estimate of the true probabilities, which provides a baseline to aim for.}\label{fig.L1Med_n}
		\end{figure}
		
		\begin{figure}[H]
			\begin{tikzpicture}[scale=0.95]
			\pgfmathsetmacro{\d}{max((pi),((e)^3)/2)}
			\pgfmathsetmacro{\C}{(3*sqrt(\d / (2*pi*e)))}
			\pgfmathsetmacro{\D}{(((e)^3/(2*pi)))}
			\pgfmathsetmacro{\k}{2200}
			\pgfmathsetmacro{\n}{1000}
			\pgfmathsetmacro{\epsnew}{(\D*\n/(2*e)+(ln(\C*\k)))/\n}
			
			\pgfmathsetmacro{\plotmarkernew}{sqrt(2*\epsnew)}
			\pgfmathsetmacro{\plotmarkerold}{sqrt(\k * 2 * ln(2) / \n)}

			\begin{axis}[legend pos=south west, axis y line = left, axis x line = bottom,title={$\frac{nC_0}{e}+2 \leq k$, $C_0=(\frac{e^3}{2\pi}) \approx 3.1967$, $C_1  \approx 2.9290$},x label style={at={(axis description cs:0.5,-0.15)},anchor=north}, xlabel={$\epsilon$}, ylabel={$\log \left(\mathbb{P}\left(\|\hat{P}_{n,k}-P\|_1 \geq \epsilon \right)\right)$}, domain = 0:2,
			enlarge y limits={rel=0.17}, xtick = \empty, extra x ticks = {0,\plotmarkernew,\plotmarkerold,2}, extra x tick labels = {$0$,\hspace{-2cm}$\sqrt{\frac{4}{n\varphi(\pi_P)}\left(\frac{nC_0}{2e}+\log(C_1k)\right)}$,$\sqrt{\frac{4k\log (2)}{n\varphi(\pi_P)}}$,$2$},ytick = {0},scaled y ticks = false]

			\addplot+{min(0,\k*ln(2)-\n/2*(x*x)};
			\addlegendentry{Lemma \ref{thm.L1Old} \cite[Theorem 2.1]{weissman2003inequalities}}
			\addplot+{min(0,\n*\epsnew-\n/2*(x*x))};
			\addlegendentry{Theorem \ref{thm.L1FromKl}}
			\end{axis}
			\end{tikzpicture}
			\includegraphics[scale=0.60]{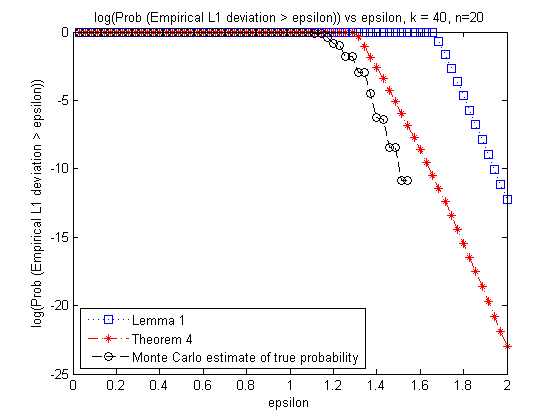}
			\caption{Cartoon and numerical plots- $\log \left(\mathbb{P}\left(\|\hat{P}_{n,k}-P\|_1 \geq \epsilon \right)\right)$ vs $\epsilon$ for small $n$.\newline When $n$ is smaller to $k$, we plot cartoon versions of the previously known bounds from Lemma \ref{thm.L1Old} and our bounds from Theorem \ref{thm.L1FromKl}. The latter shows a significant improvement, and we corroborate this plot with a numerical plot for $k=40$, $n=20$ with $P$ uniform (and hence $\varphi(\pi_P) = 2$). In the numerical figure we also plot a Monte Carlo estimate of the true probabilities, which provides a baseline to aim for. We observe that the bounds from Theorem \ref{thm.L1FromKl} are closer to the true probabilities compared with that from Lemma \ref{thm.L1Old}.	}\label{fig.L1Small_n}
		\end{figure}
		
	\end{subsection}

	\begin{section}{Future directions}\label{sec.FutDir}
	
We hope this work opens more doors than it closes. Motivated by equation~(\ref{eqn.Meanleqk/n}), we conjecture the following bound, which appears non-trivial to prove or disprove for general $n$ and $k$:
			\begin{conjecture}\label{conj.noncentral}
				\begin{align}\label{eqn.newConjec}
				\mathbb{P}\left( D(\hat{P}_{n,k} \| P) \geq \epsilon \right) \leq \left( 1 + \frac{k-1}{n}\right)^n 2e^{-n\epsilon}. 
				\end{align}
			\end{conjecture}

Note that we can't hope to get $\epsilon_{\text{thresh}}$ values much better than those implied by Conjecture ~\ref{conj.noncentral} because as $n$ grows, for $k = o(n)$, the $\epsilon_{\text{thresh}}$ we get from Conjecture ~\ref{conj.noncentral} goes to $\Theta\left(\frac{k}{n}\right)$ which we know is a lower bound on $\epsilon_{\text{thresh}}$.

In fact, with the fact that $2n \cdot D(\hat{P}_{n,k} \| P)$ is asymptotically $\chi_{k-1}^2$ from Theorem \ref{thm:VarIsk/nsq}, we can make a conjecture for the centralized concentration using known results for sub-exponential random variables. 
			\begin{conjecture}
			There exist two constants $g_1>0$, $g_2>0$ such that for any $t>0$,   
				\begin{align}\label{eqn.newCentralConjec}
					\mathbb{P}\left(|D(\hat{P}_{n,k} \| P)-\mathbb{E}[D(\hat{P}_{n,k} \| P)]| \geq t \right) \leq g_1 e^{-g_2 \min\left\{ \frac{n^2t^2}{k-1}, nt \right \}}. 
				\end{align}
			\end{conjecture}

	Moreover, we note that the $L_1$ deviation inequalities depend on $P$, while the KL divergence deviation bounds are uniform. It might be interesting to look at how to get better KL divergence deviation bounds that depend on some parameter of the distribution $P$.
	
	After initial dissemination of this work, in follow up work \cite[Thm 1.2]{agrawal2019concentration} provided another KL divergence concentration bound. Their bound performs better than Theorem \ref{thm:main} when
	\[\frac{k-1}{n} \left(\log2 + 1\right) < \epsilon <  \frac{k-1}{n} \left(\log2 + \sqrt{\frac{n}{k}} \right).\] Moreover, although they do not state this, their result immediately implies sub-Gaussian concentration for the square root of the KL divergence $\left(\sqrt{D(\hat{P}_{n,k} \| P)} \right)$ for a restricted range of parameters. Their result states that for $\epsilon > \frac{k-1}{n} \left(\log2 + 1\right)$, 
	\[\mathbb{P}\left( D(\hat{P}_{n,k} \| P) \geq \epsilon \right) \leq  e^{-n\epsilon}\left(2e\left(\frac{\epsilon n}{k-1}-\log 2\right)\right)^{k-1}  .\]
	Using the facts that $\frac{x}{2}-\log (12x) \geq 0$ if $x \geq 10$ and that $\left(2e\left(\frac{\epsilon n}{k-1}-\log 2\right)\right) < 12\frac{\epsilon n}{k-1}$, we can see that their result implies that for $\epsilon > 10 \frac{k-1}{n}$
	\[\mathbb{P}\left( D(\hat{P}_{n,k} \| P) \geq \epsilon \right) \leq  e^{-\frac{n\epsilon}{2}}.\]
	Since both the total variation ($L_1$) distance and Hellinger distance are upper bounded by a constant times the square root of the KL divergence, this result implies sub-Gaussian concentration for the empirical distribution in both these distance metrics for a restricted range of parameters. It would be interesting to see if one can obtain such a scaling for the whole range of parameters.

	\end{section}

	\begin{section}{Acknowledgments}
		The authors would like to thank Rohit Agrawal for pointing out a typo in a previous version of this paper; in Theorem~\ref{thm:VarIsk/nsq} we had missed a factor of $2$. JM would like to thank Tavor Baharav for catching a bug in an earlier version of Lemma 2.
	\end{section}
		
	\bibliographystyle{alpha}
	\bibliography{DistTypesBib}	
	
	\appendix

	\begin{section}{Proof of Theorem~\ref{thm:VarIsk/nsq}}\label{sec.VarThm}
	In this Section, we prove Theorem \ref{thm:VarIsk/nsq}. In Subsection \ref{subsec.Varleq} we prove Equation \ref{eqn.Varleqk/nsq} which says that \[\Var\left[D(\hat{P}_{n,k}\|P)\right] \leq \min\left(\frac{6\cdot (3+\log k)^2}{n},C\frac{k}{n^2}\right),\]  and in Subsection \ref{subsec.Vargeq} we prove Equation \ref{eqn.Vargeqk/nsq} which says that asymptotically \[2nD(\hat{P}_{n,k}\|P) \xrightarrow{\mathcal{D}} \chi^2_{k-1}\] and \[2(k-1) \leq \lim\limits_{n \rightarrow \infty} 4n^2 D(\hat{P}_{n,k}\|P)].\] Hence we have an upper bound and an asymptotic lower bound on $\Var\left[D(\hat{P}_{n,k}\|P)\right]$.
	\begin{subsection}{Showing $\Var\left[D(\hat{P}_{n,k}\|P)\right] \leq \min\left(\frac{6\cdot (3+\log k)^2}{n},C\frac{k}{n^2}\right)$}\label{subsec.Varleq}
		To prove Equation \ref{eqn.Varleqk/nsq} in Theorem \ref{thm:VarIsk/nsq}, our goal is to upper bound
		\[
		\Var\left[D(\hat{P}_{n,k}\|P)\right] = \Var \left( \sum_{i\in [k]} \hat{p}_i \log \frac{\hat{p}_i}{p_i} \right) \ .
		\]
		In what follows in Appendix \ref{sec.VarThm}, $c$ denotes a constant independent of $k$ and $n$, though its value might change from line to line.
		
		\begin{subsubsection}{$\Var\left[D(\hat{P}_{n,k}\|P)\right] \leq \frac{6(\log k +3)^2}{n}$}\label{subsubsec.Varleqlogk}
			We first show that 
			\begin{align}
			\Var\left[D(\hat{P}_{n,k}\|P)\right] & \leq \frac{6(\log k +3)^2}{n}.
			\end{align}
			Indeed, decomposing
			\begin{align*}
			\Var\left[D(\hat{P}_{n,k}\|P)\right] & = \Var\left(\sum_{i\in [k]} - \hat{p}_i \log \frac{1}{\hat{p}_i} + \sum_{i\in [k]} \hat{p}_i \log \frac{1}{p_i}\right) \leq 2 \left(\Var\left(\sum_{i\in [k]} - \hat{p}_i \log \frac{1}{\hat{p}_i}\right) + \Var\left(\sum_{i\in [k]} \hat{p}_i \log \frac{1}{p_i}\right)\right). 
			\end{align*}
			The first term is the plug-in entropy estimator, whose variance was shown in \cite[Lemma 15]{jiao2017maximum} to be upper bounded by 
			\[
			\Var \left( \sum_{i\in [k]} - \hat{p}_i \log \frac{1}{\hat{p}_i} \right) \leq \frac{2(\log k +3)^2}{n}. 
			\]
			Regarding the second term, since each term is increasing respect to $\hat{p}_i$, by the negative association property of multinomial distribution (There are many references for negative association properties. For example, one could consult \cite[Definition 2.1, Property P6, and Section 3.1(a)]{joag1983negative}) we obtain
			\begin{align*}
			\Var \left( \sum_{i\in [k]} \hat{p}_i \log \frac{1}{p_i} \right) & \leq \sum_{i\in [k]} \frac{p_i(1-p_i)}{n}\log^2(1/p_i) \leq \sum_{i\in [k]} \frac{p_i}{n}\log^2 (1/p_i) \leq \sum_{i\in [k]} \frac{p_i}{n}(1+\log (1/p_i))^2 \\
			& = \frac{k}{n} \sum_{i\in [k]} \frac{1}{k} p_i (1+\log (1/p_i))^2 \\
			& \leq \frac{(1+\log k)^2}{n},
			\end{align*}
			where in the last step we used the fact that $x(1+\ln(1/x))^2$ is a concave function on $(0,1]$ along with Jensen's inequality.
			\\ These two variance upper bounds complete the proof of the fact $\Var\left[D(\hat{P}_{n,k}\|P)\right] \leq \frac{6(\log k +3)^2}{n}$.
		\end{subsubsection}
		
		\begin{subsubsection}{$\Var\left[D(\hat{P}_{n,k}\|P)\right] \leq C\frac{k}{n^2}$}\label{subsubsec.Varleqk}
			We now show the other part of Equation \ref{eqn.Varleqk/nsq}, that there exists some absolute constant $C$
			\begin{align}\label{eqn.VarleqPart2}
			\Var\left[D(\hat{P}_{n,k}\|P)\right] & \leq C\frac{k}{n^2}.
			\end{align}
			We rewrite the KL-divergence as
			\begin{align*}
			\sum_{i\in [k]} \hat{p}_i \log \frac{\hat{p}_i}{p_i} & = \sum_{p_i \geq 1/n} \left( \frac{(\hat{p}_i - p_i)^2}{p_i} + \hat{p}_i - p_i  + \Delta_i \right) + \sum_{p_i <1/n} \left( \hat{p}_i \log \frac{\hat{p}_i}{p_i}  + \hat{p}_i - p_i + p_i - \hat{p}_i \right) \\
			& = \sum_{p_i \geq 1/n} \frac{(\hat{p}_i-p_i)^2}{p_i} + \sum_{p_i \geq 1/n} \Delta_i + \sum_{p_i <1/n} \left( \hat{p}_i \left( \log \frac{\hat{p}_i}{p_i}\right)  +p_i - \hat{p}_i \right)  \ ,
			\end{align*}
			where $\Delta_i = \hat{p}_i \log \frac{\hat{p}_i}{p_i} - \frac{(\hat{p}_i - p_i)^2}{p_i} + p_i - \hat{p}_i = \hat{p}_i \log (\hat{p}_i /p_i) - \hat{p}_i^2/p_i + \hat{p}_i$.
			
			Hence
			\begin{align}
			\Var\left[D(\hat{P}_{n,k}\|P)\right] & = \Var \left( \sum_{p_i \geq 1/n} \frac{(\hat{p}_i-p_i)^2}{p_i} + \sum_{p_i \geq 1/n} \Delta_i + \sum_{p_i <1/n} \left( \hat{p}_i \left( \log \frac{\hat{p}_i}{p_i}\right)  +p_i - \hat{p}_i \right) \right) \nonumber \\
			& \label{eqn.V1V2V3decomp} \leq 9 \left( \underbrace{\Var \left( \sum_{p_i \geq 1/n} \frac{(\hat{p}_i-p_i)^2}{p_i} \right)}_{:=V_1} + \underbrace{\Var \left( \sum_{p_i \geq 1/n} \Delta_i \right)}_{:=V_2} + \underbrace{\Var \left( \sum_{p_i <1/n} \left( \hat{p}_i \left( \log \frac{\hat{p}_i}{p_i}\right)  +p_i - \hat{p}_i \right) \right)}_{:=V_3} \right) \ .
			\end{align}
			
			We bound the terms $V_1,V_2$ and $V_3$ separately.
			
			\vspace{0.2cm}
			\textbf{Bounding $V_1$:}
			
			We first state the following result which we use here and prove in Lemma \ref{lemma.exactquadratic} in Appendix \ref{sec.SumToIntegral}.
			\\Let $\hat{P}_{n,k} = (\hat{p}_1,\hat{p}_2,\ldots,\hat{p}_k)$ be the empirical distribution. Denote 
			\begin{align*}
			X_i & = \frac{(\hat{p}_i - p_i)^2}{p_i} - \frac{1-p_i}{n}. 
			\end{align*}
			Then, for any $i\neq j$, 
			\begin{align*}
			\mathbb{E}[X_i^2] & = \frac{(1-p_i)(1+2(n-3)p_i(1-p_i))}{n^3 p_i} \\
			\mathbb{E}[X_iX_j] & = \frac{2(p_i + p_j) -1 + 2(n-3)p_i p_j}{n^3}. 
			\end{align*}
			Now defining $I = \{i : np_i \geq 1 \}$ we use Lemma \ref{lemma.exactquadratic} as follows.
			It is clear that $\sum_{i\in I} p_i \leq 1$ and $\sum_{i\in I} 1 \leq k$.
			\\ Now, using the fact that for the terms we are concerned with $(i \in I)$, $np_i > 1$, that $(1-p_i)<1$ and $n-3 < n$, we can write, using Lemma \ref{lemma.exactquadratic} that
			
			\begin{align}\label{eqn.XiSqEasy}
			\mathbb{E}[X_i^2] & \leq \frac{3np_i}{n^3 p_i} = \frac{3}{n^2}.
			\end{align}
			Similarly, using Lemma \ref{lemma.exactquadratic} and the facts that $p_ip_j < p_i$, $n-2<n$, $-1<0$, we get
			\begin{align*}
			\mathbb{E}[X_iX_j] & \leq 2\frac{p_j+np_i}{n^3}.
			\end{align*}
			Using these two inequalities, we can now easily upper bound $V_1$
			\begin{align*}
			V_1 & = \sum_{i,j\in I} \mathbb{E}[X_i X_j] = \sum_{i\in I} \mathbb{E}[X_i^2] + \sum_{i\in I} \sum_{j\neq i} \mathbb{E}[X_i X_j] \\
			& \leq \sum_{i \in I} \frac{3}{n^2} + \sum_{i\in I} \sum_{j\neq i} 2\frac{p_j+np_i}{n^3} \leq \frac{3k}{n^2} +   \sum_{i\in I}  2\frac{1+nkp_i}{n^3} \\
			& \leq \frac{3k}{n^2} +  2\frac{k+nk}{n^3} \leq \frac{7k}{n^2}. 
			\end{align*}
			\begin{align}\label{eqn.V1Bounded}
			V_1 \leq \frac{7 k }{n^2}. 
			\end{align}

			\vspace{0.2cm}
			\textbf{Bounding $V_2$:}
			\\We want to upper bound $V_2 = \Var \left( \sum_{p_i \geq 1/n} \Delta_i \right)$ where $\Delta_i = \hat{p}_i \log (\hat{p}_i /p_i) - \hat{p}_i^2/p_i + \hat{p}_i$. The first step we take towards this is to obtain the following upper bound -
			\begin{align}\label{eqn.V2BoundDecompose}
			V_2 = \Var \left( \sum_{p_i \geq 1/n} \Delta_i \right) \leq c\sum_{p_i \geq 1/n} \mathbb{E} \left[ \Delta_i^2 \right].
			\end{align}
			Recall that throughout this Section, $c$ is an absolute constant independent of $n$ and $k$ but may change from line to line. Equation \ref{eqn.V2BoundDecompose} is saying is that we can decompose the variance of the sum in $V_2$ into a sum of squared expectations.
			\\The standard way of obtaining such an inequality is via the Negative Association properties \cite{joag1983negative} of multinomial random variables. In particular, we know that multinomial random variables are negatively associated, disjoint monotone functions of negatively associated random variables are negatively associated, and the variance of sums of negatively associated random variables is subadditive. However, the hitch is that $\Delta_i$'s are not monotone in $\hat{p}_i$. Hence our strategy will be to decompose $\Delta_i$ as the sum of monotone functions (in $\hat{p}_i$), and use the properties of negatively associated random variables on those to decompose the variance of a sum into a sum of variances.
			\\We use the notation $f(x,y) = x\log (x/y) +x -x^2 /y$. Then
			\[
			\frac{\partial}{\partial x} f(x,y) = \log (x/y) +2 -2x/y = \log (1+z) -2z \ ,
			\]
			with $z=(x/y)-1$. The derivative is zero at $z=0 \iff x=y$ and at some $z^* \in (-1,0)$ satisfying $\log (1+z^* ) = 2z^*$, in which case $x=(1+z^*)y$. Hence the function $f(x,y)$ is decreasing in $x$ on the interval $[0,(1+z^*)y]$, increasing on the interval $[(1+z^*)y,y]$, and decreasing on $[y,1]$. Also note that $f(0,y)=f(y,y)=0$. Observe that $\Delta_i = f(\hat{p}_i ,p_i)$.
			
			Hence the decomposition $\Delta_i = \Delta^{(1)}_i + \Delta^{(2)}_i + \Delta^{(3)}_i$ with
			\begin{align*}
			\Delta^{(1)}_i & = \1 \{ \hat{p}_i \leq (1+z^*)p_i \} \left( f(\hat{p}_i ,p_i) - f((1+z^*)p_i ,p_i ) \right), \  \\
			\Delta^{(2)}_i & = \1 \{ \hat{p}_i \leq (1+z^*)p_i \} f((1+z^*)p_i ,p_i ) + \1 \{ \hat{p}_i \in [(1+z^*)p_i ,p_i ] \} f(\hat{p}_i ,p_i), \  \\
			\Delta^{(3)}_i & = \1 \{ \hat{p}_i > p_i \} f(\hat{p}_i ,p_i) 
			\end{align*}
			holds. Furthermore, $\Delta^{(1)}_i ,\Delta^{(2)}_i$ and $\Delta^{(3)}_i$ are all monotonic in $\hat{p}_i$.
			We make the following two observations that \[\mathbb{E}[\Delta^{(1)}_i\Delta^{(3)}_i] = \mathbb{E}[\Delta^{(2)}_i\Delta^{(3)}_i] = 0,\] because the support of $\Delta^{(3)}_i$ is disjoint from the supports of $\Delta^{(1)}_i$ and $\Delta^{(2)}_i$. Also,
			\[\mathbb{E}[\Delta^{(1)}_i\Delta^{(2)}_i] = \mathbb{E}[\Delta^{(1)}_i] \mathbb{E}[\Delta^{(2)}_i],\] because wherever $\Delta^{(1)}_i$ is non-zero, $\Delta^{(2)}_i$ is a constant and hence can be pulled out of the expectation.
			\\ Using these two observations, and after some algebra, we can obtain the following inequality
			\begin{align}\label{eqn.V2CrossTermsDontMatter}
			\Var \left(\Delta^{(1)}_i \right)+\Var \left(\Delta^{(2)}_i \right)+\Var \left(\Delta^{(3)}_i \right) \leq \mathbb{E}[\Delta_i^2].
			\end{align}
			Thus using the negative association of $(\hat{p}_1,\dots ,\hat{p}_k)$, the fact that $\Delta^{(1)}_i ,\Delta^{(2)}_i$ and $\Delta^{(3)}_i$ are all monotonic in $\hat{p}_i$, and Equation \ref{eqn.V2CrossTermsDontMatter}, we get
			\begin{align*}
			V_2 &\leq 3 \left( \Var \left( \sum_{p_i\geq 1/n} \Delta^{(1)}_i \right) + \Var \left( \sum_{p_i\geq 1/n} \Delta^{(2)}_i \right) + \Var \left( \sum_{p_i\geq 1/n} \Delta^{(3)}_i \right) \right) \nonumber \\
			&  \leq 3 \left( \sum_{p_i\geq 1/n} \Var \left( \Delta^{(1)}_i \right) + \sum_{p_i\geq 1/n} \Var \left( \Delta^{(2)}_i \right) + \sum_{p_i\geq 1/n} \Var \left( \Delta^{(3)}_i \right) \right) \nonumber \\
			& \leq 3\sum_{p_i \geq 1/n} \mathbb{E} \left[ \Delta_i^2 \right].
			\end{align*}
			which proves Equation \ref{eqn.V2BoundDecompose}.
			\\ To further bound this quantity, we introduce the following lemma.
			\begin{lemma}
				For every $x \geq -1$, these standard logarithmic inequalities
				\[ x \leq (1+x)\log (1+x) \leq (1+x) \left( x \right)
				\]
				hold.
			\end{lemma}
			\begin{proof}
				For $x=-1$, the inequality is true because $-1\leq 0 \leq 0$. Hence now suppose $x \in \left(-1,1\right]$.
				We then have that \[1+x > 0\] so dividing on throughout by $1+x$, we only need to show that \[\underbrace{\frac{x}{1+x} - \log (1+x)}_{f_1(x)} \leq 0 \leq \underbrace{ x - \log (1+x)}_{f_2(x)}\] for all $x \in \left(-1,1\right]$. To do this, observe that $f_1(0)=0$ ($f_2(0)=0$). If we can show that \[f_1'(x) \leq 0 \hspace{0.3cm} (0 \leq f_2'(x)) \text{ for } 0 \leq x \leq 1\]  and \[0 \leq f_1'(x) \hspace{0.3cm}(f_2'(x) \leq 0) \text{ for } -1 < x \leq 0,\] then we are done. Computing the derivative shows that these conditions hold because \[f_1'(x) = \frac{-x}{(1+x)^2}\] and \[f_2'(x) = \frac{x}{1+x}.\] This concludes the proof of the lemma.
				
			\end{proof}
			
			Using the Lemma with $x=\tfrac{\hat{p}_i}{p_i}-1$ yields
			\[\left( \frac{\hat{p}_i}{p_i}-1  \right) \leq 
			\frac{\hat{p}_i}{p_i} \log \frac{\hat{p}_i}{p_i} \leq \frac{\hat{p}_i}{p_i} \left( \frac{\hat{p}_i}{p_i}-1  \right) \ .
			\]
			Rearranging and multiplying throughout by $p_i$ yields
			\begin{align*}
			-\frac{(\hat{p}_i-p_i)^2}{p_i} \leq \Delta_i  \leq 0.
			\end{align*}
			Hence,
			\begin{align*}
			\mathbb{E}[\Delta_i^2] & \leq \mathbb{E}\left[ \left(\frac{(\hat{p}_i-p_i)^2}{p_i}\right)^2  \right]. 
			\end{align*}
			Now, using Lemma \ref{lemma.exactquadratic} and the relaxation in Equation \ref{eqn.XiSqEasy}, we obtain
			\begin{align*}
			\mathbb{E}[\Delta_i^2] & \leq 2\left( \frac{3}{n^2}+\frac{(1-p_i)^2}{n^2}\right) \leq \frac{8}{n^2}.
			\end{align*}
			Plugging this bound into the inequality $V_2 \leq 3\sum_{p_i \geq 1/n} \mathbb{E} \left[ \Delta_i^2 \right]$ we obtain
			\begin{align}\label{eqn.V2Bounded}
			V_2 \leq \frac{c k }{n^2}. 
			\end{align}
			
			\vspace{0.2cm}
			\textbf{Bounding $V_3$:}
			\\To upper bound $V_3$, we again split into a sum of monotone functions and use negative association of multinomial random variables. Note that
			\begin{align*}
			V_3 & = \Var \left( \sum_{p_i <1/n} \left( \hat{p}_i \log \frac{\hat{p}_i}{p_i}   \right) + \sum_{p_i <1/n}  \left(p_i - \hat{p}_i\right) \right) \leq 2\left( \Var \left( \sum_{p_i <1/n} \left( \hat{p}_i \log \frac{\hat{p}_i}{p_i}   \right)\right)  + \Var \left(\sum_{p_i <1/n}  \left(p_i - \hat{p}_i\right) \right)  \right).
			\end{align*}
			$p_i-\hat{p}_i$ is monotone decreasing in $\hat{p}_i$ and hence we can upper bound the variance of the sum by the sum of the variance using negative association. To do the same for the $\hat{p}_i \log \frac{\hat{p}_i}{p_i}$ term, note that
			\[
			\frac{\partial}{\partial x} \left( x\log \frac{x}{p} \right) = \log \frac{x}{p} + 1 \ ,
			\]
			which is positive when $x>p$. Since $p_i < 1/n$, whenever $\hat{p}_i \geq \frac{1}{n}$, we have $\hat{p}_i \log \frac{\hat{p}_i}{p_i}$ is increasing in $\hat{p}_i$. It is also clear that $0 = 0\log \frac{0}{p} < \frac{1}{n} \log \frac{1}{np}$. Since the variance computation only cares about the values of $\hat{p}_i \log \frac{\hat{p}_i}{p_i}$ at $\hat{p}_i=\frac{l}{n}$ for $l \in \{0,1,...,n\}$ and we have shown that the function is monotone increasing when restricted to this set of inputs, we can again use negative association to upper bound the variance of the sum by the sum of the variance. This gives us		\begin{align}\label{eqn.BoundingV3}
			V_3 \leq 2 \left( \sum_{p_i <1/n} \Var \left( \hat{p}_i \log \frac{\hat{p}_i}{p_i}  \right) + \sum_{p_i <1/n} \Var \left( p_i - \hat{p}_i\right) \right). 
			\end{align}

			Note that $\Var \left( p_i - \hat{p}_i\right) = \Var (\hat{p}_i ) = \tfrac{p_i (1-p_i )}{n} < 1/n^2$ when $p_i <1/n$. On the other hand
			\begin{align*}
			\mathbb{E} \left[ \hat{p}_i^2 \log^2 \frac{\hat{p}_i}{p_i} \right] & = \sum_{j=1}^n {n \choose j} p_i^j (1-p_i)^{n-j} \frac{j^2}{n^2} \log^2 \frac{j}{np_i} \\
			& \leq \frac{1}{n^2} \sum_{j=1}^n \left( \frac{e}{j} \right)^j j^2 \log^2 j \\
			& \leq \frac{c}{n^2} \ ,
			\end{align*}
			using that ${n \choose j} \leq (en/j)^j$ and $np_i <1$, and that $\max \limits_{x \in (0,1)} \left(x \log^2 \frac{1}{x} \right) < 1$.
			\\ Using these bounds and the fact that $\Var \left( \hat{p}_i \log \frac{\hat{p}_i}{p_i}  \right) \leq \mathbb{E} \left[ \hat{p}_i^2 \log^2 \frac{\hat{p}_i}{p_i} \right]$, and that the summation in Equation \ref{eqn.BoundingV3} has at most $k$ terms, we conclude that
			\begin{align}\label{eqn.V3Bounded}
			V_3 \leq \frac{c k }{n^2}. 
			\end{align}

			Having shown that for some constant $c$ we have $V_i \leq c \frac{k}{n^2}$ for $i=1,2,3$ (Equations \ref{eqn.V1Bounded}, \ref{eqn.V2Bounded}, \ref{eqn.V3Bounded}), we now use Equation \ref{eqn.V1V2V3decomp} to complete the proof of Equation \ref{eqn.VarleqPart2} which finishes Subsubsection \ref{subsubsec.Varleqk}.
		\end{subsubsection}

		Combining the results of Subsubsections \ref{subsubsec.Varleqlogk} and \ref{subsubsec.Varleqk} immediately completes the proof of Equation \ref{eqn.Varleqk/nsq} in Theorem \ref{thm:VarIsk/nsq} which finishes Subsection \ref{subsec.Varleq}.	
	\end{subsection}
	
	\begin{subsection}{Showing $2(k-1) \leq \lim\limits_{n \rightarrow \infty} 4n^2 \Var\left[D(\hat{P}_{n,k}\|P)\right]$}\label{subsec.Vargeq}
		The idea behind Equation \ref{eqn.Vargeqk/nsq} in Theorem \ref{thm:VarIsk/nsq} is a use of the well known delta method in statistics. Here we need a second order multivariate version of the delta method. One can consult \cite[Chapter 3]{van2000asymptotic} for an introduction to the Delta method. The idea behind the delta method is that to understand the asymptotic distribution of a functional of a random variable whose asymptotic distribution we understand, one can use a Taylor approximation of the functional to the desired precision. We want to show the following.
		\\ For a fixed $k$, for any $ P \in \mathcal{M}_k$, asymptotically as $n$ goes to infinity,
		\begin{align*}
		2nD(\hat{P}_{n,k}\|P) \xrightarrow{\mathcal{D}} \chi^2_{k-1},
		\\ 2(k-1) \leq \lim\limits_{n \rightarrow \infty} 4n^2 \Var\left[D(\hat{P}_{n,k}\|P)\right].
		\end{align*}
		Here $\chi^2_{k-1}$ is the chi-square distribution with $k-1$ degrees of freedom and hence has variance $2(k-1)$. If we have \[2nD(\hat{P}_{n,k}\|P) \xrightarrow{\mathcal{D}} \chi^2_{k-1},\] then using Slutsky's Theorem \cite[2.8]{van2000asymptotic} and the Continuous Mapping theorem \cite[2.3]{van2000asymptotic}, we obtain \[\left(2nD(\hat{P}_{n,k}\|P)-\mathbb{E}[2nD(\hat{P}_{n,k}\|P)]\right)^2 \xrightarrow{\mathcal{D}} \left(\chi^2_{k-1} - \mathbb{E}[\chi^2_{k-1}]\right)^2.\] Keeping in mind that \[\Var\left(\chi^2_{k-1}\right) = 2(k-1),\] the asymptotic variance lower bound in Theorem \ref{thm:VarIsk/nsq} follows from this by applying Fatou's Lemma \cite[18.13]{carothers2000real} on the sequence of positive random variables \[\left(2nD(\hat{P}_{n,k}\|P)-\mathbb{E}[2nD(\hat{P}_{n,k}\|P)]\right)^2.\]
		\\ Hence we now only need to show that $2nD(\hat{P}_{n,k}\|P) \xrightarrow{\mathcal{D}} \chi^2_{k-1}$. \vspace{0.5cm}
		\\ First, let us compute a Taylor series approximation of $D(Q\|P)$.
		To do that, we need to parametrize $P$ with $k-1$ variables. Hence, let us consider \[P=(p_1,p_2,...,p_{k-1},1-\sum_{j=1}^{k-1}p_j)\] with parameters $p_1,...,p_{k-1}$. Define \[p_k = 1-\sum_{j=1}^{k-1}p_j,\] and given a $(k \times 1)$ vector $V$, let $V_{-k}$ denote the $(k-1 \times 1)$ vector obtained from $V$ by removing the $k^{th}$ component. We need the first and second order derivatives of \[D(Q\|P) = \sum_{i=1}^{k-1} q_i\log \frac{q_i}{p_i} + q_k\log \frac{q_k}{1-\sum_{j=1}^{k-1}p_j}\] at $P$. Here $Q = (q_1,q_2,...,q_k)$.
		Let	\[G \triangleq \nabla_{Q_{-k}} D(Q\|P)|_{Q_{-k}=P_{-k}}\] be a $(k-1 \times 1)$  vector and denote its $i^{th}$ component by $\nabla_{Q_{-k}} D(Q\|P)|_{Q_{-k}=P_{-k}}(i)$
		Let \[H \triangleq \nabla_{Q_{-k}}^2 D(Q\|P)|_{Q_{-k}=P_{-k}}\] be a $(k-1 \times k-1)$  matrix and denote its $(i,j)^{th}$ entry by $\nabla_{Q_{-k}}^2 D(Q\|P)|_{Q_{-k}=P_{-k}}(i,j)$.
		For notation, let $\vec{0}$ be the $(k-1 \times 1)$ all-zeros vector, $\vec{1}$ be the $(k-1 \times 1)$ all-ones vector, \[P_{-k}^{-1} \triangleq (\frac{1}{p_1},\frac{1}{p_2},...,\frac{1}{p_{k-1}}),\] and $\mathsf{diag}(P_{-k}^{-1})$ be the $(k-1 \times k-1)$ matrix with $P_{-k}^{-1}$ on its diagonal, and $\mathcal{I}_{k-1}$ be the $(k-1 \times k-1)$ identity matrix.
		\\ We can calculate that
		\begin{align*}
		G(i) = \nabla_{Q_{-k}} D(Q\|P)|_{Q_{-k}=P_{-k}}(i) = \log \frac{q_i}{p_i}-\log\frac{1-\sum_{j=1}^{k-1}q_j}{1-\sum_{j=1}^{k-1}p_j} = \log \frac{p_i}{p_i}-\log\frac{1-\sum_{j=1}^{k-1}p_j}{1-\sum_{j=1}^{k-1}p_j} = 0.
		\end{align*}
		So, \[G = \nabla_{Q_{-k}} D(Q\|P)|_{Q_{-k}=P_{-k}} = \vec{0}.\]
		\begin{align*}
		H(i,j)=\nabla_{Q_{-k}}^2 D(Q\|P)|_{Q_{-k}=P_{-k}}(i,j) = \begin{cases}
		\frac{1}{q_k} = \frac{1}{p_k}& i \neq j \\
		\frac{1}{q_i}+\frac{1}{q_k} = \frac{1}{p_i}+\frac{1}{p_k}& i = j
		\end{cases}.
		\end{align*}
		So, \[H = \nabla_{Q_{-k}}^2 D(Q\|P)|_{Q_{-k}=P_{-k}} = \frac{1}{p_k}\vec{1}\vec{1}^T + \mathsf{diag}(P_{-k}^{-1}).\]
		\\ For the rest of Subsection \ref{subsec.Vargeq}, denote the empirical distribution with $n$ draws $\hat{P}_{n,k}$ by $\hat{P}$, suppressing the subscripts which are obvious by context for this subsection.
		\begin{lemma}\label{lem.MultiCLT}
			\cite[2.18]{van2000asymptotic} By the multivariate central limit theorem, we have that 
			\begin{align*}
			\sqrt{n}\left(\hat{P}_{-k}-P_{-k}\right) \xrightarrow{\mathcal{D}} \mathcal{N}\left(\vec{0},\Sigma\right),
			\end{align*}
			where $\Sigma$ is a $(k-1 \times k-1)$ matrix with
			\begin{align*}
			\Sigma(i,j) = \begin{cases}
			-p_i p_j & i \neq j \\
			p_i(1-p_i) & i = j
			\end{cases}.
			\end{align*}
		\end{lemma}
		\begin{observation}\label{obs.SigmaHisI}
			\begin{align*}
			\Sigma H = \mathcal{I}_{k-1}
			\end{align*}
		\end{observation}
		Now we can write our Taylor Series expansion as
		\begin{align*}
		nD(\hat{P}_{n,k}\|P) = n\underbrace{D(P\|P)}_{0} + \underbrace{nG^T\left(\hat{P}_{-k}-P_{-k}\right)}_{0} + \underbrace{\frac{1}{2}\sqrt{n}\left(\hat{P}_{-k}-P_{-k}\right)^TH\sqrt{n}\left(\hat{P}_{-k}-P_{-k}\right)}_{:=\text{Quadratic term}} + \text{ Higher order terms}.
		\end{align*}
		Let $\mathcal{Z}$ be a random variable that is distributed as $\mathcal{N}\left(\vec{0},\Sigma\right)$. Because of the fact that quadratic (and cubic and higher order) maps are continuous, we can use the Continuous Mapping theorem along with Lemma \ref{lem.MultiCLT} to get
		\begin{align*}
		\text{Quadratic term} \xrightarrow{\mathcal{D}} \frac{1}{2}\mathcal{Z}^TH\mathcal{Z}.
		\end{align*}
		Because the higher order terms all contain cubics or higher powers of $\left(\hat{P}_{-k}-P_{-k}\right)$, each of which are only pre-multiplied by an $n$, we again have (because of Lemma \ref{lem.MultiCLT}) , by the Continuous mapping theorem -
		\begin{align*}
		\text{Higher order terms} \xrightarrow{\mathcal{D}} 0.
		\end{align*}
		Now, using Slutsky's Theorem, we get
		\begin{align*}
		nD(\hat{P}_{n,k}\|P) \xrightarrow{\mathcal{D}} \frac{1}{2}\mathcal{Z}^TH\mathcal{Z} + 0 = \frac{1}{2}\mathcal{Z}^TH\mathcal{Z}.
		\end{align*}
		So all that remains to be done to complete the proof of Equation \ref{eqn.Vargeqk/nsq} of Theorem \ref{thm:VarIsk/nsq} is to show that $\mathcal{Z}^TH\mathcal{Z}$ is distributed as $\chi^2_{k-1}$. Quadratic forms of multivariate Gaussian random variables are well studied and we have the following result from \cite[Chapter 29 (Quadratic Forms in Normal Variables)]{johnson1970distributions} which we state here for the reader's convenience.
		\begin{lemma}\label{lem.QuadNormalChiSq}
			If $\mathcal{Z}$ is a $m \times 1$ multivariate Gaussian Random variable with mean $\vec{0}$ and nonsingular covariance matrix $\mathcal{V}$, then the quadratic form $\mathcal{Z}^TA\mathcal{Z}$ is distributed as
			\begin{align*}
			\sum \limits_{j=1}^m \lambda_j W_j^2,
			\end{align*}
			where $\lambda_1 \geq \lambda_2 \geq ... \geq \lambda_m$ are the eigenvalues of $\mathcal{V}A$ and $W_j$'s are i.i.d. $\mathcal{N}\left(0,1\right)$ random variables.
		\end{lemma}
		Applying Lemma \ref{lem.QuadNormalChiSq} to our quantity of interest $\mathcal{Z}^TH\mathcal{Z}$, coupled with Observation \ref{obs.SigmaHisI} (which says that all the $k-1$ eigenvalues of $\Sigma H$ are $1$) says that $\mathcal{Z}^TH\mathcal{Z}$ is distributed as $\sum \limits_{j=1}^{k-1} W_j^2$, which is exactly the definition of a $\chi^2_{k-1}$ random variable. This concludes our proof of the fact that $2nD(\hat{P}_{n,k}\|P) \xrightarrow{\mathcal{D}} \chi^2_{k-1}$, and concludes the proof of Theorem \ref{thm:VarIsk/nsq}.
	\end{subsection}
	
	\begin{subsection}{Variance lower bounds in the Poissonized model}\label{subsec.Poisson}
		In this section we prove Theorem \ref{thm.Poi} which demonstrates why it was necessary to handle the dependencies and cancellations in Theorem \ref{thm:VarIsk/nsq} as we did rather than simply use the Poissonization technique to work with independent random variables which are much easier to work with.
		We have $\hat{P}_{n,k}^{\mathsf{Poi}} = \left(\hat{p}_1^{\mathsf{Poi}},\hat{p}_2^{\mathsf{Poi}},...,\hat{p}_k^{\mathsf{Poi}}\right)$ where each $\hat{p}_i^{\mathsf{Poi}}$ is independently distributed as $\frac{\mathsf{Poi}(np_i)}{n}$, where $\mathsf{Poi}(\lambda)$ is a Poisson random variable with parameter $\lambda$. We want to show that
		\begin{align*}
		\sqrt{n}D(\hat{P}_{n,k}^{\mathsf{Poi}}\|P) \xrightarrow{\mathcal{D}} \mathcal{N}\left(0,1\right),
		\\ 1 \leq \lim\limits_{n \rightarrow \infty} n \Var\left[D(\hat{P}_{n,k}^{\mathsf{Poi}}\|P)\right]. \nonumber
		\end{align*}
		Just like in Subsection \ref{subsec.Vargeq}, $1 \leq \lim\limits_{n \rightarrow \infty} n \Var\left[D(\hat{P}_{n,k}^{\mathsf{Poi}}\|P)\right]$ follows from $\sqrt{n}D(\hat{P}_{n,k}^{\mathsf{Poi}}\|P) \xrightarrow{\mathcal{D}} \mathcal{N}\left(0,1\right)$ using a combination of Slutsky's Theorem, the Continuous mapping theorem, and Fatou's Lemma. Hence we only need to show \[\sqrt{n}D(\hat{P}_{n,k}^{\mathsf{Poi}}\|P) \xrightarrow{\mathcal{D}} \mathcal{N}\left(0,1\right).\]
		The key difference between the lower bound in this model and in the multinomial model is that for the Poissonized model, we do not have a constraint of the form $\sum_{i=1}^{k} \hat{p}_i^{\mathsf{Poi}} = 1$, which we did in the multinomial model in Subsection \ref{subsec.Vargeq}. This means we use $k$ parameters in our Taylor series expansion for the delta method, and most importantly, the first order derivative does NOT vanish.
		We first make the following two observations -
		\begin{observation}\label{obs.PoiToNorm}
			As $n \rightarrow \infty$, \[\frac{\mathsf{Poi}(np_i)-np_i}{\sqrt{np_i}} \xrightarrow{\mathcal{D}} \mathcal{N}\left(0,1\right).\]
		\end{observation}
		\begin{proof}
			We can view $\mathsf{Poi}(np_i)$ as a sum of $n$ i.i.d. $\mathsf{Poi}(p_i)$ random variables because of the properties of Poisson random variables. The statement then follows from a simple application of the Central Limit Theorem.
		\end{proof}
		\begin{observation}\label{obs.PoiDerivative}
			$\nabla_{Q} D(Q\|P)(i) = 1+\log \frac{q_i}{p_i} \Rightarrow \nabla_{Q} D(Q\|P)|_{Q=P} = \vec{1} $ 
		\end{observation}
	
		Armed with these two observations, we write the Taylor series expansion -
		\[\sqrt{n}D(\hat{P}_{n,k}^{\mathsf{Poi}}\|P) = \sqrt{n}\underbrace{D(P||P)}_{0} + \sqrt{n} \left(\nabla_{Q} D(Q\|P)|_{Q=P}\right)^T \left(\hat{P}_{n,k}^{\mathsf{Poi}}-P\right) + \text{Second and higher order terms}.\]
		Using Observation \ref{obs.PoiToNorm} and the fact that 'Second and higher order terms' have quadratic or higher powers in $\left(\hat{P}_{n,k}^{\mathsf{Poi}}-P\right)$ only premultiplied by a $\sqrt{n}$, we obtain that $\text{Second and higher order terms} \xrightarrow{\mathcal{D}} 0$. Using Observations \ref{obs.PoiToNorm} and \ref{obs.PoiDerivative}, we get \[\sqrt{n} \left(\nabla_{Q} D(Q\|P)|_{Q=P}\right)^T \left(\hat{P}_{n,k}^{\mathsf{Poi}}-P\right) \xrightarrow{\mathcal{D}} \mathcal{N}\left(0,1\right).\] Using Slutsky's Theorem, this completes the proof of \[\sqrt{n}D(\hat{P}_{n,k}^{\mathsf{Poi}}\|P) \xrightarrow{\mathcal{D}} \mathcal{N}\left(0,1\right),\] which completes the proof of Theorem \ref{thm.Poi} and ends Subsection \ref{subsec.Poisson}.
	\end{subsection}
	
\end{section}
	
	\begin{section}{Proof of Theorem~\ref{thm:main}}\label{sec.maintheoremproof}
	
	\begin{subsection}{Using conditional probabilities and chain rule to reduce the size $k$ problem to a size $k-1$ problem}\label{sec.Conditioning}
				Now we begin our proof of Theorem \ref{thm:main}.
				In what follows, $\tilde{P}$ will be a distribution that may change from line to line and may even mean two different things in the same line. The only important thing is that it is a distribution and is of the right support size for the context it appears in.
		\\ Suppose $\hat{p}_k \neq 1$, then by the chain rule of relative entropy we have, with a little algebra

		\begin{align}
		D(\hat{P}_{n,k} \| P) & = \sum_{i=1}^{k} \hat{p}_i\log\frac{\hat{p}_i}{p_i} \nonumber
		\\ &\label{eqn.2Decomposition} = \underbrace{D((\hat{p}_k,1-\hat{p}_k) \| (p_k,1-p_k))}_{:=A_{\hat{p}_k}} + (1-\hat{p}_k)\underbrace{D(\hat{P}_{n(1-\hat{p}_k),k-1} \| \tilde{P})}_{:=B_{\hat{p}_k}},
		\end{align}
		where $\hat{P}_{n(1-\hat{p}_k),k-1}$ and $\tilde{P}$ are both distributions in $\mathcal{M}_{k-1}$ and defined as 
		\[\hat{P}_{n(1-\hat{p}_k),k-1} \triangleq (\frac{\hat{p}_1}{1-\hat{p}_k},\frac{\hat{p}_2}{1-\hat{p}_k},\ldots,\frac{\hat{p}_{k-1}}{1-\hat{p}_k}),\]
		\[\tilde{P} \triangleq (\frac{p_1}{1-p_k},\frac{p_2}{1-p_k},\ldots,\frac{p_{k-1}}{1-p_k}).\]
		The idea now is to control the probability of $D(\hat{P}_{n,k} \| P)$ being large by conditioning on the value of $\hat{p}_k$. This fixes $A_{\hat{p}_k}$ and $1-\hat{p}_k$ and we can control $B_{\hat{p}_k}$ by using our control on the probability of $D(\hat{P}_{\cdot,k-1} \| \tilde{P})$ (by building up inductively).
		
		The random variables $\hat{p}_1,\hat{p}_2,\ldots,\hat{p}_k$ are not independent of each other.

	However, by the law of total probability, since $\hat{p}_k$ can take values in $\{\frac{0}{n},\frac{1}{n},...,\frac{n-1}{n},\frac{n}{n}\}$

		\begin{align}
			& \mathbb{P}\left( D(\hat{P}_{n,k} \| P) \geq \epsilon \right) \nonumber \\
			&\quad  = \sum\limits_{l=0}^{n} \mathbb{P}\left( D(\hat{P}_{n,k} \| P) \geq \epsilon | n\hat{p}_k=l \right) \mathbb{P}\left(n\hat{p}_k=l \right) \nonumber
			\\ & \quad = \sum\limits_{l=0}^{n-1} \mathbb{P}\left( D(\hat{P}_{n,k} \| P) \geq \epsilon | n\hat{p}_k=l \right) \mathbb{P}\left(n\hat{p}_k=l \right) + \mathbb{P}\left( D(\hat{P}_{n,k} \| P) \geq \epsilon | n\hat{p}_k=n \right) \mathbb{P}\left(n\hat{p}_k=n \right) \nonumber
			\\ & \quad = \sum\limits_{l=0}^{n-1} \mathbb{P}\left( B_{\hat{p}_k} \geq \frac{\epsilon-A_{\hat{p}_k}}{1-\hat{p}_k} | n\hat{p}_k=l \right) \mathbb{P}\left(n\hat{p}_k=l \right) + \mathbb{P}\left( D(\hat{P}_{n,k} \| P) \geq \epsilon | n\hat{p}_k=n \right) \mathbb{P}\left(n\hat{p}_k=n \right) \nonumber
			\\ &\label{eqn.2Conditioning} \quad = \sum\limits_{l=0}^{n-1} \mathbb{P}\left( D(\hat{P}_{n-l,k-1} \| \tilde{P}) \geq \frac{\epsilon-A_{\frac{l}{n}}}{1-\frac{l}{n}}  \right) \mathbb{P}\left(n\hat{p}_k=l \right) + \underbrace{\mathbb{P}\left( D(\hat{P}_{n,k} \| P) \geq \epsilon | n\hat{p}_k=n \right) \mathbb{P}\left(n\hat{p}_k=n \right)}_{:=T}.
		\end{align}
		In the third equality, since for $l \in [n-1]$ we have $\hat{p}_k \neq 1$, we have used Equation \ref{eqn.2Decomposition}.
		\\ We will plug in an upper bound for the $k-1$ sized problem into Equation \ref{eqn.2Conditioning} to do our computations. 
		\\ 
		\begin{definition}\label{defn.EpsSet}
			Let $\mathcal{E} \subseteq [n]$ be the set containing all $l \in [n]$ such that $A_{\frac{l}{n}} > \epsilon$.
		\end{definition}
		\begin{align}
		&	\mathbb{P}\left( D(\hat{P}_{n,k} \| P) \geq \epsilon \right) \nonumber \\
		 &\quad  \leq \sum\limits_{l \in [n] \setminus\mathcal{E}}^{} \mathbb{P}\left( D(\hat{P}_{n-l,k-1} \| \tilde{P}) \geq \frac{\epsilon-A_{\frac{l}{n}}}{1-\frac{l}{n}}  \right) \mathbb{P}\left(n\hat{p}_k=l \right) + \sum\limits_{l \in \mathcal{E}}^{} 1 \cdot \mathbb{P}\left(n\hat{p}_k=l \right) + T \nonumber
			\\ & \quad = \sum\limits_{l \in [n] \setminus\mathcal{E}}^{} \mathbb{P}\left( D(\hat{P}_{n-l,k-1} \| \tilde{P}) \geq \frac{\epsilon-A_{\frac{l}{n}}}{1-\frac{l}{n}}  \right) \mathbb{P}\left(n\hat{p}_k=l \right) + \mathbb{P}\left(A_{\hat{p}_k} > \epsilon \right) + T \nonumber
			\\ &\label{eqn.2ConditioningBetter} \quad = \underbrace{\sum\limits_{l \in [n] \setminus\mathcal{E}}^{} \mathbb{P}\left( D(\hat{P}_{n-l,k-1} \| \tilde{P}) \geq \frac{\epsilon-A_{\frac{l}{n}}}{1-\frac{l}{n}}  \right) \mathbb{P}\left(n\hat{p}_k=l \right)}_{:=R_k} + \underbrace{\mathbb{P}\left(D((\hat{p}_k,1-\hat{p}_k) \| (p_k,1-p_k)) > \epsilon \right)}_{:=S} + T
		\end{align}
		In the first equality we have used the fact that KL divergence is always non-negative..
		\\ For $S$ we can simply use our tight characterization of the distribution of types for $k=2$ and from Example \ref{eg:TightFork=2} observe that \[\mathbb{P}\left( D(\hat{P}_{n,k} \| P) \geq \epsilon \right) \geq \mathbb{P}\left(D((\hat{p}_k,1-\hat{p}_k) \| (p_k,1-p_k)) > \epsilon \right)\] could be as large as $2 e^{-n\epsilon}$ because all the terms in Equation \ref{eqn.2ConditioningBetter} are non-negative. 
		\begin{observation}\label{obs.SIsSmall}
			$S \leq 2e^{-n\epsilon}$, and it could be as large as $2e^{-n\epsilon}$ using Example \ref{eg:TightFork=2}.
		\end{observation}

		\begin{observation}\label{obs.MaxD}
			Define \[p_{\text{min}} \triangleq \min\limits_{i \in [k]} p_i.\] Then \[\max\limits_{\text{empirical distributions } \hat{P}_{n,k}} D(\hat{P}_{n,k} \| P) \leq \max\limits_{Q \in \mathcal{M}_k} D(Q \| P) = \log\frac{1}{p_{\text{min}}}.\] We can see this because $D(Q \| P)$ is convex in $Q$ over a compact convex set, $\mathcal{M}_k$ and hence must attain its maxima at an extreme point of $\mathcal{M}_k$. This means that the maxima must be attained at a distribution which puts all its mass in one spot. Of all these $k$ distributions, the one which puts all its mass on $p_{\text{min}}$ maximizes $D(Q \| P)$. 
		\end{observation}

		So far we have been agnostic to what the actual distribution $P$ is. Since $\mathbb{P}\left( D(\hat{P}_{n,k} \| P) \geq \epsilon \right)$ is invariant under permutations of the support set, we might as well compute this quantity with any permutation of our choice. So assume that $p_{\text{min}} = p_k$. This means we choose to condition on the value taken by the outcome which has least probability.
		\\ 
		Using Observation \ref{obs.MaxD}, we get that \[\mathbb{P}\left( D(\hat{P}_{n,k} \| P) \geq \epsilon \right) = 0,\] if $\epsilon \geq \log\frac{1}{p_k}$. Hence we only need to consider the following for  $\epsilon \leq \log\frac{1}{p_k}$.
		\begin{align}\label{eqn.BoundT}
			T & = \mathbb{P}\left( D(\hat{P}_{n,k} \| P) \geq \epsilon | n\hat{p}_k=n \right) \mathbb{P}\left(n\hat{p}_k=n \right) \nonumber
			\\ & = \mathbb{P}\left( \log \frac{1}{p_k} \geq \epsilon \right) \cdot p_k^n \leq 1 \cdot p_k^n = e^{-n\log\frac{1}{p_k}} \nonumber
			\\ & \leq e^{-n\epsilon}.
		\end{align}
		This holds for all $\epsilon$ because when $\epsilon \geq \log\frac{1}{p_k}$, $T=0$.
		\\
		\\
		We have bounded both $S$ and $T$, and now turn to $R_k$ from Equation \ref{eqn.2ConditioningBetter}.
	\end{subsection}

	\begin{subsection}{Using $k=3$ to understand the behaviour of $R_k$}\label{sec.k=3,4}	
		Let $k=3$. Using Equation \ref{eqn.binaryalphabet} we have, using the definition of $A_{\frac{l}{n}}$ from Equation \ref{eqn.2Decomposition}-
		\begin{align}
			R_3 & = \sum\limits_{l \in [n] \setminus\mathcal{E}}^{} \mathbb{P}\left( D(\hat{P}_{n-l,k-1} \| \tilde{P}) \geq \frac{\epsilon-A_{\frac{l}{n}}}{1-\frac{l}{n}}  \right) \mathbb{P}\left(n\hat{p}_k=l \right) \nonumber
			\\ &\label{eqn.StopHereForLargek} = \sum\limits_{l \in [n] \setminus\mathcal{E}}^{} \mathbb{P}\left( D(\hat{P}_{n-l,2} \| \tilde{P}) \geq \frac{\epsilon-A_{\frac{l}{n}}}{1-\frac{l}{n}}  \right) \mathbb{P}\left(n\hat{p}_k=l \right)
			\\ & \leq \sum\limits_{l \in [n] \setminus\mathcal{E}}^{} 2e^{-(n-l)(\frac{\epsilon-A_{\frac{l}{n}}}{1-\frac{l}{n}})}\mathbb{P}\left(n\hat{p}_k=l \right) \nonumber
			\\ & = \sum\limits_{l \in [n] \setminus\mathcal{E}}^{} 2e^{-n\epsilon}\{e^{nD((\frac{l}{n},1-\frac{l}{n}) \| (p_k,1-p_k))}\mathbb{P}\left(n\hat{p}_k=l \right)\} \nonumber
			\\ &\label{eqn.RelaxedExp} \leq \sum\limits_{l=0}^{n} 2e^{-n\epsilon}\{e^{nD((\frac{l}{n},1-\frac{l}{n}) \| (p_k,1-p_k))}\mathbb{P}\left(n\hat{p}_k=l \right)\}
			\\ & = 2e^{-n\epsilon} \underbrace{\mathbb{E}_{X \sim \mathsf{B}(n,p_k)}[e^{nD((\frac{X}{n},1-\frac{X}{n}) \| (p_k,1-p_k))}]}_{:=E} \nonumber
			\\ & \leq 2e^{-n\epsilon} (\frac{e}{2} \sqrt{n}) = 2e^{-n\epsilon} \frac{ec_0}{2\pi} \sqrt{n} \nonumber.
		\end{align}
		In the last inequality we need the expectation $E$ which we upper bound in Lemma \ref{lem.UpperBoundEk} in Section \ref{sec.ExpComputation} and use above.
		\\ Hence using Equation \ref{eqn.2ConditioningBetter} and our results for $R_3$ (above), $S$ (Observation \ref{obs.SIsSmall}) and $T$ (Equation \ref{eqn.BoundT}), we get 
		\begin{align}\label{eqn.k=3Result}
		\boxed{\mathbb{P}\left( D(\hat{P}_{n,3} \| P) \geq \epsilon \right) \leq 2e^{-n\epsilon} (\frac{e}{2} \sqrt{n}) + 3e^{-n\epsilon} = e^{-n\epsilon}\left[3\left(1 + K_0\frac{e\sqrt{n}}{2\pi}\right)\right].} 
		\end{align}
		Here $c_0$ and $K_0$ are as defined in Equations \ref{eqn.cmDefn} and \ref{eqn.KmDefn}, and this result gives Theorem \ref{thm:main} for the special case $k=3$ (Note that the statement of Theorem \ref{thm:main} has an extra factor of $\frac{c_1}{c_2}$ for reasons that will become clearer later).

	\end{subsection}

	\begin{subsection}{The case for general $k$}\label{sec.Largek}
		In what follows, $c_m$ and $K_m$ are defined as in Equations \ref{eqn.cmDefn} and \ref{eqn.KmDefn}. Also, here we define the following 
		\begin{align}\label{eqn.hmDefn}
			h_m \triangleq \begin{cases}
			c_m & m \neq 2 \\
			c_1 & m=2 
			\end{cases},
		\end{align}
		\begin{align}\label{eqn.HmDefn}
			H_m \triangleq \prod_{j=0}^{m} h_j.
		\end{align}
		We will also later use the fact that
		\begin{align}\label{eqn.KandHreln}
			H_m \leq \frac{c_1}{c_2}K_m  \text{ for } \ m \geq 0.
		\end{align}		
		
		 We have shown in Equation \ref{eqn.k=3Result} that \[\mathbb{P}\left( D(\hat{P}_{n,k} \| P) \geq \epsilon \right) \leq e^{-n\epsilon}\left[3\sum_{i=0}^{k-2}H_{i-1} (\frac{e\sqrt{n}}{2\pi})^{i}\right]\] holds for $k=3$. So we will induct assuming it is true for $k-1$, and hope to show it for value $k$.
		
		Using Equation \ref{eqn.StopHereForLargek}, we have, using the definition of $A_{\frac{l}{n}}$ from Equation \ref{eqn.2Decomposition}-
		\begin{align}
			R_k & = \sum\limits_{l \in [n] \setminus\mathcal{E}}^{} \mathbb{P}\left( D(\hat{P}_{n-l,k-1} \| \tilde{P}) \geq \frac{\epsilon-A_{\frac{l}{n}}}{1-\frac{l}{n}}  \right) \mathbb{P}\left(n\hat{p}_k=l \right) \nonumber
			\\ & \leq \sum\limits_{l=0}^{n} \mathbb{P}\left( D(\hat{P}_{n-l,k-1} \| \tilde{P}) \geq \frac{\epsilon-A_{\frac{l}{n}}}{1-\frac{l}{n}}  \right) \mathbb{P}\left(n\hat{p}_k=l \right) \nonumber
			\\ &\label{eqn.UsingDiffBoundsCareful} \leq \sum\limits_{l=0}^{n} e^{-(n-l)(\frac{\epsilon-A_{\frac{l}{n}}}{1-\frac{l}{n}})}\left[3\sum_{i=0}^{k-3}H_{i-1}  (\frac{e\sqrt{n-l}}{2\pi})^{i}\right]\mathbb{P}\left(n\hat{p}_k=l \right)
			\\ & = e^{-n\epsilon}\left[3 \sum_{i=0}^{k-3}H_{i-1}(\frac{e\sqrt{n}}{2\pi})^i \left(\sum_{l=0}^{n} \left(\sqrt{1-\frac{l}{n}}\right)^i e^{nA_{\frac{l}{n}}}\mathbb{P}\left(n\hat{p}_k=l \right) \right) \right] \nonumber
			\\ &\label{eqn.RelaxedExpForLargek} =
			e^{-n\epsilon}\left[3 \sum_{i=0}^{k-3}H_{i-1}(\frac{e\sqrt{n}}{2\pi})^i \underbrace{\mathbb{E}_{X \sim \mathsf{B}(n,p_k)}\left[\left(\sqrt{1-\frac{X}{n}}\right)^i e^{nD((\frac{X}{n},1-\frac{X}{n}) \| (p_k,1-p_k))}\right]}_{:=E_{i}} \right]
			\\ & \leq e^{-n\epsilon}\left[3 \sum_{i=0}^{k-3}H_{i-1}(\frac{e\sqrt{n}}{2\pi})^i \left(h_i \frac{e\sqrt{n}}{2\pi}\right) \right] \nonumber
			\\ & \leq e^{-n\epsilon}\left[3 \sum_{i=1}^{k-2}H_{i-1}(\frac{e\sqrt{n}}{2\pi})^i \right] \nonumber.
		\end{align}
		
		In Inequality \ref{eqn.UsingDiffBoundsCareful}, we have used our inductive assumption \[\mathbb{P}\left( D(\hat{P}_{n,k-1} \| P) \geq \epsilon \right) \leq e^{-n\epsilon}\left[3\sum_{i=0}^{k-3}H_{i-1} (\frac{c_0}{c_2}\frac{e\sqrt{n}}{2\pi})^{i}\right],\] and we have upper bounded the expectation $E_i$ by $\left(h_i \frac{e\sqrt{n}}{2\pi}\right)$ in Lemma \ref{lem.UpperBoundEk} in Section \ref{sec.ExpComputation}.
		\\ Hence using Equation \ref{eqn.2ConditioningBetter} and our upper bounds for $R_k$ (above), $S$ (Observation \ref{obs.SIsSmall}) and $T$ (Equation \ref{eqn.BoundT}), we get (with Equation \ref{eqn.KandHreln})
		\begin{align*}
		\mathbb{P}\left( D(\hat{P}_{n,3} \| P) \geq \epsilon \right) & \leq e^{-n\epsilon}\left[3 \sum_{i=1}^{k-2}H_{i-1}(\frac{e\sqrt{n}}{2\pi})^i \right]+ 3e^{-n\epsilon} \\
		&  = e^{-n\epsilon}\left[3 \sum_{i=0}^{k-2}H_{i-1}(\frac{e\sqrt{n}}{2\pi})^i \right] \\
		& \leq e^{-n\epsilon}\left[\frac{3c_1}{c_2} \sum_{i=0}^{k-2}K_{i-1}(\frac{e\sqrt{n}}{2\pi})^i \right].
		\end{align*}
		This proves our inductive step, and using upper bounds on $K_m$ from Equation \ref{eqn.KmDefn}, this completes the proof of Equation \ref{eqn.ResultAllk} in Theorem \ref{thm:main}. To get more interpretable versions of this bound, we do the following.
		\\ Using the fact $K_m \leq \sqrt{\frac{d_0}{m}}\left(\sqrt{\frac{2\pi e}{m}}\right)^m$ in conjunction with this result, we obtain \[\mathbb{P}\left( D(\hat{P}_{n,k} \| P) \geq \epsilon \right) \leq e^{-n\epsilon}\left[\frac{3c_1}{c_2}\sqrt{\frac{d_0}{2\pi e}} \left(\sum_{i=1}^{k-2}\left(\sqrt{\frac{e^3n}{2\pi i}}\right)^{i}+1\right)\right].\]
		Recalling that $C_1=\frac{3c_1}{c_2}\sqrt{\frac{d_0}{2\pi e}}$ and $C_0=(\frac{e^3}{2\pi})$, the rest of the piecewise bounds in Theorem \ref{thm:main} follow straightforwardly after making the following observations-
		\begin{itemize}
			\item If $k \leq \sqrt{nC_0}+2$, then $k \leq \frac{n}{k} \implies k\left(\sqrt{\frac{C_0n}{k}}\right)^{k-2} \leq \left(\sqrt{\frac{C_0n}{k}}\right)^{k}$.
			\item The function $\left(\sqrt{\frac{C_0n}{i}}\right)^i$ is maximized at $i = \frac{C_0n}{e}$.
			\item The function $\left(\sqrt{\frac{C_0n}{i}}\right)^i \leq 1$ for $i \geq C_0n$.
			\item $\left(\frac{k}{k-2}\right)^{k-2} \leq e^2$ for $k \geq 3$.
		\end{itemize}

		\begin{align*}
			& \mathbb{P}\left( D(\hat{P}_{n,k} \| P) \geq \epsilon \right)
			\\ & \leq e^{-n\epsilon}\left[\frac{3c_1}{c_2}\sum_{i=0}^{k-2}K_{i-1} (\frac{e\sqrt{n}}{2\pi})^{i}\right] \leq e^{-n\epsilon}\left[\frac{3c_1}{c_2}\sqrt{\frac{d_0}{2\pi e}} \left(\sum_{i=1}^{k-2}\left(\sqrt{\frac{e^3n}{2\pi i}}\right)^{i}+1\right)\right]	
			\\ & \leq \begin{cases} 
			e^{-n\epsilon}\left[\frac{3c_1}{c_2}\sqrt{\frac{d_0}{2\pi e}}(k-2) \left(\sqrt{\frac{e^3n}{2\pi (k-2)}}\right)^{k-2}\right]\leq \boxed{C_1e\left(\sqrt{\frac{C_0n}{k}}\right)^ke^{-n\epsilon}} & 3\leq k \leq \sqrt{nC_0}+2 \\
			e^{-n\epsilon}\left[\frac{3c_1}{c_2}\sqrt{\frac{d_0}{2\pi e}}(k-2) \left(\sqrt{\frac{e^3n}{2\pi (k-2)}}\right)^{k-2}\right]\leq \boxed{C_1k\left(\sqrt{\frac{C_0n}{k}}\right)^ke^{-n\epsilon}} & 3 \leq k\leq \frac{nC_0}{e}+2 \\
			e^{-n\epsilon}\left[\frac{3c_1}{c_2}\sqrt{\frac{d_0}{2\pi e}}(k-2) e^{\frac{n(\frac{e^3}{2\pi})}{2e}}\right]\leq \boxed{C_1k e^{\frac{C_0n}{2e}}e^{-n\epsilon}} & \frac{nC_0}{e}+2 \leq k\leq nC_0+2 \\	
			e^{-n\epsilon}\left[\frac{3c_1}{c_2}\sqrt{\frac{d_0}{2\pi e}}\left(nC_0 e^{\frac{n(\frac{e^3}{2\pi})}{2e}} + k-2-nC_0\right)\right] \leq \boxed{C_1\left(nC_0 e^{\frac{C_0n}{2e}}+k\right)e^{-n\epsilon}} & k\geq nC_0+2
			\end{cases}
			\end{align*}
		This completes the proof of Theorem \ref{thm:main}.
		\end{subsection}
	
	\begin{subsection}{Computing the expectation of the exponential of the KL Divergence in the binary alphabet}\label{sec.ExpComputation}
	Note: Because we use $0\log0 = 0$, we also use $0^0 = 1$. \\
	
	In this section, we prove the following upper bound on $E_i$ (defined in Equation \ref{eqn.RelaxedExpForLargek}) which has been used in Sections \ref{sec.k=3,4} and \ref{sec.Largek}.
	\begin{lemma}\label{lem.UpperBoundEk}
		\begin{align*}
		E_i = \mathbb{E}_{X \sim \mathsf{B}(n,p_k)}\left[\left(\sqrt{1-\frac{X}{n}}\right)^i e^{nD((\frac{X}{n},1-\frac{X}{n}) \| (p_k,1-p_k))}\right] \leq h_i\frac{e\sqrt{n}}{2\pi}
		\end{align*}
	\end{lemma}

	\begin{proof}
		Below, we use the following Stirling Approximation that is valid for all integers $n$ on each of the $3$ factorials involved in $\binom{n}{l}$.
		\begin{align}
		\sqrt{2\pi}n^{n+\frac{1}{2}}e^{-n} \leq n! \leq en^{n+\frac{1}{2}}e^{-n}. 
		\end{align}
		\begin{align*}
		E_i & =	\mathbb{E}_{X \sim \mathsf{B}(n,p_k)}\left[\left(\sqrt{1-\frac{X}{n}}\right)^i e^{nD((\frac{X}{n},1-\frac{X}{n}) \| (p_k,1-p_k))}\right]
		\\ & = \sum\limits_{l=0}^{n} \binom{n}{l}p_k^l(1-p_k)^{n-l} e^{n(\frac{l}{n}\log\frac{l}{np_k}+\frac{n-l}{n}\log\frac{n-l}{n(1-p_k)})}\left(\sqrt{1-\frac{l}{n}}\right)^i
		\\ & = \sum\limits_{l=0}^{n} \binom{n}{l}p_k^l(1-p_k)^{n-l} \frac{l^l}{n^l p_k^l} \frac{(n-l)^{n-l}}{n^{n-l}(1-p_k)^{n-l}}\left(\sqrt{1-\frac{l}{n}}\right)^i
		\\ & = \sum\limits_{l=0}^{n} \binom{n}{l} \frac{l^l (n-l)^{n-l}}{n^n}\left(\sqrt{1-\frac{l}{n}}\right)^i
		\\ & \leq  \sum\limits_{l=0}^{n} \frac{en^{n+\frac{1}{2}}e^{-n}}{\sqrt{2\pi}l^{l+\frac{1}{2}}e^{-l}\sqrt{2\pi}(n-l)^{n-l+\frac{1}{2}}e^{-(n-l)}} \frac{l^l (n-l)^{n-l}}{n^n}\left(\sqrt{1-\frac{l}{n}}\right)^i
		\\ & = \frac{e}{2\pi} \sum\limits_{l=0}^{n} \frac{\sqrt{n}}{\sqrt{l}\sqrt{n-l}}\left(\sqrt{1-\frac{l}{n}}\right)^i
		\\ & = \frac{e\sqrt{n}}{2\pi} \sum\limits_{l=0}^{n} \frac{\left(\sqrt{1-\frac{l}{n}}\right)^i}{\sqrt{\frac{l}{n}}\sqrt{1-\frac{l}{n}}}\frac{1}{n}
		\end{align*}
		For all non-negative integers $i\neq 2$, we can now use Lemmas \ref{lem.SumToIntegral} and \ref{lem.TheseFuncsAreConvex} from Appendix \ref{sec.SumToIntegral} and certain definite integrals from Section \ref{sec.Integrals} to upper bound this sum.
		\begin{align*}
			E_i \leq \frac{e\sqrt{n}}{2\pi} \sum\limits_{l=0}^{n} \frac{\left(\sqrt{1-\frac{l}{n}}\right)^i}{\sqrt{\frac{l}{n}}\sqrt{1-\frac{l}{n}}}\frac{1}{n} \leq \frac{e\sqrt{n}}{2\pi}\int_{0}^{1} \frac{(1-x)^{\frac{i}{2}}}{\sqrt{x-x^2}}dx = c_i\frac{e\sqrt{n}}{2\pi} \leq h_i\frac{e\sqrt{n}}{2\pi}
		\end{align*}
		For $i=2$, we observe that $E_2 \leq E_1 \leq c_1\frac{e\sqrt{n}}{2\pi} = h_2\frac{e\sqrt{n}}{2\pi}$.
		This completes the proof of Lemma \ref{lem.UpperBoundEk}.
	\end{proof}

\end{subsection}

\end{section}

	\begin{section}{Auxiliary lemmas}\label{sec.SumToIntegral}

		\begin{lemma}
			\label{lem:DiffSlopeBound}
			For any distribution $P \in \mathcal{M}_k$ and any subset $F \subseteq \mathcal{M}_k$, given $n$ iid samples from $P$ with $\hat{P}_{n,k} =(\hat{p}_1,\hat{p}_2,\ldots,\hat{p}_k)$ denoting the empirical distribution, we have
			\begin{align}\label{eqn.DiffSlopeBound}
				\mathbb{P}\left(\hat{P}_{n,k} \in F \right) \leq 2(k-1)e^{-n\frac{\inf\limits_{P' \in F} D(P'\|P)}{k-1}}
			\end{align}
		\end{lemma}
		\begin{proof}
			First, it is clear that \[\mathbb{P}\left(\hat{P}_{n,k} \in F \right) \leq \mathbb{P}\left(D(\hat{P}_{n,k}\|P) \geq \inf\limits_{P' \in F} D(P'\|P) \right).\] Hence we only need to focus on upper bounds for $\mathbb{P}\left(D(\hat{P}_{n,k}\|P) \geq \epsilon \right)$ for some fixed $\epsilon=\inf\limits_{P' \in F} D(P'\|P)$. The structure of this proof is the same as that of Theorem \ref{thm:main}. We decompose $D(\hat{P}_{n,k}\|P)$ as in Equation \ref{eqn.2Decomposition}, use the law of total probability and then use induction. The only difference is that now we use a different inductive hypothesis and bound the terms differently. Our inductive hypothesis is that for all $\tilde{P} \in \mathcal{M}_{k-1}$ ($k \geq 3$) and all positive integers $n$ we have \[\mathbb{P}\left(D(\hat{P}_{n,k-1}\|\tilde{P}) \geq \epsilon \right) \leq 2(k-2)e^{-\frac{n\epsilon}{k-2}}.\] The base case for $k=3$ is immediate from Lemma \ref{lem.k=2Case}. Using Equation \ref{eqn.2Decomposition} we then have
			\begin{align*}
				\mathbb{P}\left(D(\hat{P}_{n,k}\|P) \geq \epsilon \right) & = \mathbb{P}\left( D(\hat{P}_{n,k} \| P) \geq \epsilon | n\hat{p}_k=n \right)\mathbb{P}\left(n\hat{p}_k=n \right) \\ &\quad + \mathbb{P}\left(D((\hat{p}_k,1-\hat{p}_k) \| (p_k,1-p_k)) + (1-\hat{p}_k)D(\hat{P}_{n(1-\hat{p}_k),k-1} \| \tilde{P}) \geq \epsilon | n\hat{p}_k\neq n\right)\mathbb{P}\left(n\hat{p}_k\neq n\right) \\
				& \leq \mathbb{P}\left(D((\hat{p}_k,1-\hat{p}_k) \| (p_k,1-p_k)) \geq q_k\epsilon | n\hat{p}_k= n\right)\mathbb{P}\left(n\hat{p}_k= n\right)\\ &\quad + \mathbb{P}\left(D((\hat{p}_k,1-\hat{p}_k) \| (p_k,1-p_k)) \geq q_k\epsilon | n\hat{p}_k\neq n\right)\mathbb{P}\left(n\hat{p}_k\neq n\right)\\ &\quad + \mathbb{P}\left((1-\hat{p}_k)D(\hat{P}_{n(1-\hat{p}_k),k-1} \| \tilde{P}) \geq (1-q_k)\epsilon | n\hat{p}_k\neq n\right)\mathbb{P}\left(n\hat{p}_k\neq n\right) \\
				& \leq  \mathbb{P}\left(D((\hat{p}_k,1-\hat{p}_k) \| (p_k,1-p_k)) \geq q_k\epsilon \right) + \mathbb{P}\left((1-\hat{p}_k)D(\hat{P}_{n(1-\hat{p}_k),k-1} \| \tilde{P}) \geq (1-q_k)\epsilon \right)\\
				& \leq 2e^{-nq_k\epsilon} + 2(k-2)e^{\frac{n(1-\hat{p}_k)(1-q_k)\epsilon}{(1-\hat{p}_k)k-2}} \\
				& \leq 2(k-1)e^{-\frac{n\epsilon}{k-1}}.
			\end{align*}
			Above we have used Lemma \ref{lem.k=2Case}, the inductive hypothesis, and eventually set $q_k=\frac{1}{k-1}$. This completes the proof of Lemma \ref{lem:DiffSlopeBound}.
		\end{proof}
		
		\begin{lemma}\label{lemma.exactquadratic}
		Let $\hat{P}_{n,k} = (\hat{p}_1,\hat{p}_2,\ldots,\hat{p}_k)$ be the empirical distribution. Denote 
		\begin{align*}
		X_i & = \frac{(\hat{p}_i - p_i)^2}{p_i} - \frac{1-p_i}{n}. 
		\end{align*}
		Then, for any $i\neq j$, 
		\begin{align*}
		\mathbb{E}[X_i^2] & = \frac{(1-p_i)(1+2(n-3)p_i(1-p_i))}{n^3 p_i}, \\
		\mathbb{E}[X_iX_j] & = \frac{2(p_i + p_j) -1 + 2(n-3)p_i p_j}{n^3}. 
		\end{align*}
	\end{lemma}
	\begin{proof}
		First, let us set up some notation. Here, $\hat{p}_i$ is distributed as $\frac{\mathsf{B}(n,p_i)}{n}$. Let $\mathsf{B}(n,p_i) = \sum_{l=1}^{n} Z_l^i$ where $Z_l^i$'s are i.i.d. Bernoulli($p_i$) random variables. Similarly, $\hat{p}_j$ is distributed as $\frac{\mathsf{B}(n,p_j)}{n}$. Let $\mathsf{B}(n,p_i) = \sum_{l=1}^{n} Z_m^j$ where $Z_m^j$'s are i.i.d. Bernoulli($p_j$) random variables. The correlation between $Z_l^i$'s and $Z_m^j$'s is the obvious one inherited from the fact that they are part of a multinomial distribution. In particular, this is their joint distribution-
		\begin{align*}
		\mathbb{P}\left((Z_l^i,Z_m^j) = (z_1,z_2) \right) = \begin{cases}
		(1-p_i-p_j)\delta_{lm} & (z_1,z_2)=(0,0) \\
		p_i\delta_{lm} & (z_1,z_2)=(1,0) \\
		p_j\delta_{lm} & (z_1,z_2)=(0,1) \\
		0 & \text{otherwise} 
		\end{cases},
		\end{align*}
		where $\delta_{lm}$ is the Kronecker delta.
		\\ We first note, using a rearrangement of terms and the fact about binomial distributions that $\mathbb{E}[(\hat{p}_i-p_i)^2] = \frac{p_i(1-p_i)}{n}$, that the two statements in Lemma \ref{lemma.exactquadratic} are equivalent to the following two statements-
		\begin{align}\label{eqn.4thcentralmomlem1}
		\mathbb{E}\left[\left(n\hat{p}_i-np_i\right)^4\right] = \mathbb{E}\left[\left(\sum_{l=1}^{n} \left(Z_l^i - p_i\right)
		\right)^4\right] = np_i(1-p_i) + 3n(n-2)p_i^2(1-p_i)^2
		\end{align}
		and
		\begin{align}\label{eqn.4thcrossmomlem1}
		\mathbb{E}\left[\left(n\hat{p}_i-np_i\right)^2\left(n\hat{p}_j-np_j\right)^2\right]  & = \mathbb{E}\left[\left(\sum_{l=1}^{n} \left(Z_l^i - p_i\right)
		\right)^2\left(\sum_{m=1}^{n}\left(Z_m^j - p_j\right)\right)^2\right] \nonumber \\ 
		& = np_ip_j\left[(n-1)-(n-2)(p_i+p_j)+(3n-6)p_ip_j\right].
		\end{align}
		We will prove both these statements by induction on $n$.
		\\ First we prove Equation \ref{eqn.4thcentralmomlem1}. The base case for $n=1$ is the following \[\mathbb{E}\left[\left(Z_1^i - p_i
		\right)^4\right] = p_i(1-p_i)^4+(1-p_i)p_i^4 = p_i(1-p_i)(1-3p_i(1-p_i)).\] We now make the inductive hypothesis that \[\mathbb{E}\left[\left(\sum_{l=1}^{n-1} \left(Z_l^i - p_i\right)
		\right)^4\right] = (n-1)p_i(1-p_i) + 3(n-1)(n-3)p_i^2(1-p_i)^2.\]
		Our inductive step then follows as
		\begin{align*}
		\mathbb{E}\left[\left(\sum_{l=1}^{n} \left(Z_l^i - p_i\right)
		\right)^4\right] & = \mathbb{E}\left[\left(\left(\sum_{l=1}^{n-1} \left(Z_l^i - p_i\right) + \left(Z_n^i-p_i\right)\right)
		\right)^4\right]\\
		& = \underbrace{\mathbb{E}\left[\left(\sum_{l=1}^{n-1} \left(Z_l^i - p_i\right)
			\right)^4\right]}_{(n-1)p_i(1-p_i) + 3(n-1)(n-3)p_i^2(1-p_i)^2} + \underbrace{4\mathbb{E}\left[\left(\sum_{l=1}^{n} \left(Z_l^i - p_i\right)
			\right)^3\left(Z_n^i-p_i\right)\right]}_{0} \\
		& \qquad + \underbrace{6\mathbb{E}\left[\left(\sum_{l=1}^{n} \left(Z_l^i - p_i\right)
			\right)^2\left(Z_n^i-p_i\right)^2\right]}_{6(n-1)p_i^2(1-p_i)^2} + \underbrace{4\mathbb{E}\left[\left(\sum_{l=1}^{n} \left(Z_l^i - p_i\right)
			\right)\left(Z_n^i-p_i\right)^3\right]}_{0}\\
		& \qquad + \underbrace{\mathbb{E}\left[\left(Z_n^i-p_i\right)^4\right]}_{p_i(1-p_i)(1-3p_i(1-p_i))} \\
		& = np_i(1-p_i) + 3n(n-2)p_i^2(1-p_i)^2.
		\end{align*}
		There are five terms of interest above in the expansion. In the $1^{st}$ term, we use our inductive hypothesis. In the $2^{nd}$, $3^{rd}$, and $4^{th}$ terms we use the fact that $Z_n^i$ is uncorrelated with all the other $Z_l^i$'s to split the expectation of the product into a product of expectations, each of which is simply a binomial random variable mean or variance computation. In the $5^{th}$ term we use our base case.
		\\This completes the proof of Equation \ref{eqn.4thcentralmomlem1}.
		\\ We now begin proving Equation \ref{eqn.4thcrossmomlem1}. The base case for $n=1$ follows as \[\mathbb{E}\left[\left(Z_1^i - p_i\right)^2\left(Z_1^j - p_j\right)^2\right] = p_i(1-p_i)^2p_j^2 + p_j(1-p_j)^2p_i^2 + (1-p_i-p_j)p_i^2p_j^2 =  p_ip_j\left[(p_i+p_j)-3p_ip_j\right],\] where we have used the joint distribution of $Z_1^i$ and $Z_1^j$ as specified above.
		We assume our inductive hypothesis that \[\mathbb{E}\left[\left(\sum_{l=1}^{n-1} \left(Z_l^i - p_i\right)
		\right)^2\left(\sum_{m=1}^{n-1}\left(Z_m^j - p_j\right)\right)^2\right] =  (n-1)p_ip_j\left[(n-2)-(n-3)(p_i+p_j)+(3n-9)p_ip_j\right].\]
		Our inductive step then follows as
		\begin{align*}
		& \mathbb{E}\left[\left(\sum_{l=1}^{n} \left(Z_l^i - p_i\right)
		\right)^2\left(\sum_{m=1}^{n}\left(Z_m^j - p_j\right)\right)^2\right]	\nonumber \\
		& = 				\mathbb{E}\left[\left(\sum_{l=1}^{n-1} \left(Z_l^i - p_i\right) + \left(Z_n^i-p_i\right)
		\right)^2\left(\sum_{m=1}^{n-1}\left(Z_m^j - p_j\right) + \left(Z_n^j-p_j\right)\right)^2\right] \\
		& = \underbrace{\mathbb{E}\left[\left(\sum_{l=1}^{n-1} \left(Z_l^i - p_i\right)
			\right)^2\left(\sum_{m=1}^{n-1}\left(Z_m^j - p_j\right)\right)^2\right]}_{(n-1)p_ip_j\left[(n-2)-(n-3)(p_i+p_j)+(3n-9)p_ip_j\right]} + \underbrace{\mathbb{E}\left[\left(\sum_{l=1}^{n-1} \left(Z_l^i - p_i\right)
			\right)^2 2\left(\sum_{m=1}^{n-1}\left(Z_m^j - p_j\right)\right)\left(Z_n^j-p_j\right)\right]}_{0}\\
		& \quad + \underbrace{\mathbb{E}\left[\left(\sum_{l=1}^{n-1} \left(Z_l^i - p_i\right)
			\right)^2\left(Z_n^j-p_j\right)^2\right]}_{(n-1)p_i(1-p_i)p_j(1-p_j)} + \underbrace{\mathbb{E}\left[\left(\sum_{m=1}^{n-1} \left(Z_m^j - p_j\right)
			\right)^2 2\left(\sum_{l=1}^{n-1}\left(Z_l^i - p_i\right)\right)\left(Z_n^i-p_i\right)\right]}_{0}\\
		& \quad + \underbrace{\mathbb{E}\left[4 \left(\sum_{m=1}^{n-1} \left(Z_m^j - p_j\right)
			\right) \left(\sum_{l=1}^{n-1}\left(Z_l^i - p_i\right)\right)\left(Z_n^i-p_i\right) \left(Z_n^j-p_j\right)\right]}_{4(n-1)(-p_ip_j)(-p_ip_j)} + \underbrace{\mathbb{E}\left[2 \left(\sum_{l=1}^{n-1}\left(Z_l^i - p_i\right)\right)\left(Z_n^i-p_i\right) \left(Z_n^j-p_j\right)^2\right]}_{0} \\
		& \quad + \underbrace{\mathbb{E}\left[\left(\sum_{m=1}^{n-1} \left(Z_m^j - p_j\right)
			\right)^2\left(Z_n^i-p_i\right)^2\right]}_{(n-1)p_j(1-p_j)p_i(1-p_i)} + \underbrace{\mathbb{E}\left[2 \left(\sum_{m=1}^{n-1}\left(Z_m^j - p_j\right)\right)\left(Z_n^j-p_j\right) \left(Z_n^i-p_i\right)^2\right]}_{0} \nonumber \\
			& \quad + \underbrace{\mathbb{E}\left[\left(Z_1^i - p_i\right)^2\left(Z_1^j - p_j\right)^2\right]}_{p_ip_j\left[(p_i+p_j)-3p_ip_j\right]}\\
		& = np_ip_j\left[(n-1)-(n-2)(p_i+p_j)+(3n-6)p_ip_j\right].
		\end{align*}
		In the nine terms of the expansion above, we have repeatedly used the fact $Z_n^i$ and $Z_n^j$ are independent of the other $Z_l^i$'s and $Z_m^j$'s to reduce expectations of products to products of expectations. The first term follows from the inductive hypothesis and the last term follows from the base case. The rest of the terms require only the knowledge of the mean and covariance matrix of a multinomial random variable, which are well known.
		\\ This completes the proof of Equation \ref{eqn.4thcrossmomlem1} and hence of Lemma \ref{lemma.exactquadratic}.
	\end{proof}
	\begin{lemma}\label{lem.k=2Case}
		Let $k=2$ and $P \in \mathcal{M}_k$ be any probability distribution. Then for all $n, \epsilon$ 
		\begin{align*}
		\mathbb{P}\left( D(\hat{P}_{n,2} \| P) \geq \epsilon \right) & \leq 2 e^{-n\epsilon}.
		\end{align*} 
	\end{lemma}
	\begin{proof}
		Let \[\hat{P}_{n,2} = (\hat{p}_1,1-\hat{p}_1), P = (p_1,1-p_1).\] $\mathcal{M}_2$ is a line segment and $\{\hat{P}_{n,2}: D(\hat{P}_{n,2} \| P)\geq \epsilon\} \in \mathcal{M}_2$ is a union of two line segments (which are, in particular, convex).
		In fact, \[\{\hat{P}_{n,2}: D(\hat{P}_{n,2} \| P)\geq \epsilon\} = \underbrace{\{\hat{P}_{n,2}: D(\hat{P}_{n,2} \| P)\geq \epsilon, \hat{p}_1 > p_1\}}_{:=M_1} \cup \underbrace{\{\hat{P}_{n,2}: D(\hat{P}_{n,2} \| P)\geq \epsilon, \hat{p}_1 < p_1\}}_{:=M_2},\] and so
		\begin{align*}
			\mathbb{P}\left( D(\hat{P}_{n,2} \| P) \geq \epsilon \right) & = \mathbb{P}\left(M_1 \cup M_2\right)
			\\ & \leq \mathbb{P}\left(M_1\right)+\mathbb{P}\left( M_2\right)
			\\ & \leq e^{-n\epsilon}+e^{-n\epsilon} = 2e^{-n\epsilon}.
		\end{align*}
		To bound $\mathbb{P}\left(M_1\right)$ and $\mathbb{P}\left(M_2\right)$ we use the fact that $M_1$ and $M_2$ are convex, $\inf\limits_{Q \in M_1} D(Q\|P) = \inf\limits_{Q \in M_2} D(Q\|P) = \epsilon$, and inequality (2.16) in \cite[Theorem 1]{csiszar1984sanov}.
	\end{proof}
	
	\begin{lemma}\label{lem.SumToIntegral}
					Let $f(x)$ be a convex function on $(0,1)$. Let $l_1<l_2<n$ be some integers.
			Then we must have 
			\begin{align*}
			\sum_{l=l_1}^{l_2}f\left(\frac{l}{n}\right)\frac{1}{n} \leq \int_{\frac{l_1}{n}}^{\frac{l_2}{n}}f(x)dx.
			\end{align*}
	\end{lemma}
	\begin{proof}
		The proof of this follows from using convexity and Jensen's Inequality to observe that \[f(\frac{l}{n}) = f(\int_{\frac{l-\frac{1}{2}}{n}}^{\frac{l+\frac{1}{2}}{n}}x \cdot n \cdot dx) \leq n\cdot \int_{\frac{l-\frac{1}{2}}{n}}^{\frac{l+\frac{1}{2}}{n}}f(x)dx,\] and then simply adding this inequality for all $l$ from $l_1$ to $l_2$.
	\end{proof}

		\begin{lemma}\label{lem.TheseFuncsAreConvex}
			$\frac{(1-x)^{\frac{m}{2}}}{\sqrt{x-x^2}}$ is convex for $x \in (0,1)$ for any non-negative integer $m$ except $m=2$.
		\end{lemma}
	
		\begin{proof}
			Let \[f(x) = \frac{(1-x)^{\frac{m}{2}}}{\sqrt{x-x^2}}\] and \[g(x) = \log f(x).\]
			Then \[f'' = (g''+g'^2)f.\] Since $f>0$ for all $x \in (0,1)$, $f$ is convex iff $(g''+g'^2)>0$ for all $x \in (0,1)$.
			\[g' = \frac{-\frac{m-1}{2}}{(1-x)}-\frac{1}{2x},g'' = \frac{-\frac{m-1}{2}}{(1-x)^2}+\frac{1}{2x^2}\] 
			So, \[g''+g'^2 = \frac{3}{4x^2} + \frac{(\frac{m-1}{2})(\frac{m-1}{2}-1)}{(1-x)^2}+\frac{(\frac{m-1}{2})\frac{1}{2}}{x(1-x)}.\]
			For $m \neq 0,2$, observe that all the terms in $g''+g'^2$ are positive. For $m=0$, one can verify that though all the terms aren't positive, the resulting function is always positive.
		\end{proof}
		
	\end{section}

	\begin{subsection}{Some useful definite integrals}\label{sec.Integrals}
		\begin{enumerate}
			
			\item
			\begin{align}\label{eqn.Integral2}
				c_m & \triangleq \int_{0}^{1} \frac{(1-x)^{\frac{m}{2}}}{\sqrt{x-x^2}}dx \\
				& = \int_{0}^{\frac{\pi}{2}}2(\cos\theta)^m d\theta \\
				& = \begin{cases}
					\frac{1 \times 3 \times 5 \ldots \times m-1}{2\times 4 \times 6 \times \ldots \times m} \cdot \pi & m\text{ is even} \\
					\frac{2\times 4 \times 6 \times \ldots \times m-1}{1 \times 3 \times 5 \ldots \times m} \cdot 2 & m\text{ is odd}
				\end{cases}
			\end{align}

			\item 
			
			\begin{align}\label{eqn.Integral3}
				\int_{-1}^{1}(\sqrt{1-x^2})^{m}dx & = 2\int_{0}^{\frac{\pi}{2}} (\cos\theta)^{m+1} d\theta \nonumber
				\\ & = c_{m+1}
			\end{align}

		\end{enumerate}
		
	\end{subsection}

\end{document}